\documentclass[conference]{IEEEtran}

\usepackage{amsmath}
\usepackage{amsfonts}
\usepackage{amsthm}
\usepackage{algorithmic}
\usepackage{graphicx}
\usepackage{textcomp}
\usepackage{threeparttable}
\usepackage{hyperref}
\usepackage{epsfig, endnotes, xurl, color, float}
\usepackage{diagbox, booktabs, colortbl}
\usepackage{listings, framed}
\usepackage{multirow, caption, subcaption}
\usepackage{xspace, slashbox, enumitem}
\DeclareMathAlphabet{\mathcal}{OMS}{cmsy}{m}{n}
\usepackage{textcomp}
\usepackage{nicematrix}
\usepackage{anyfontsize}
\usepackage{ragged2e}
\usepackage{array}
\newcolumntype{C}[1]{>{\centering\arraybackslash}p{#1}}

\mathchardef\mhyphen="2D

\newcommand\sysname{$\mathsf{ammBoost}$\xspace}
\newcommand\ammop{$\mathsf{ammOP}$\xspace}

\newcommand\sk{\mathsf{sk}\xspace}
\newcommand\pk{\mathsf{pk}\xspace}
\newcommand\vk{\mathsf{vk}\xspace}
\newcommand\ppt{\mathsf{PPT}\xspace}
\newcommand{\param}{\ensuremath{\mathtt{pp}}\xspace}
\newcommand\clients{\mathcal{C}\xspace}
\newcommand\miners{\mathcal{M}\xspace}
\newcommand\mainc{\mathsf{mc}\xspace}
\newcommand\sidec{\mathsf{sc}\xspace}
\newcommand\syssetup{\mathsf{SystemSetup}\xspace}
\newcommand\partysetup{\mathsf{PartySetup}\xspace}
\newcommand\stt{\mathsf{state}\xspace}
\newcommand\createTransaction{\mathsf{CreateTx}\xspace}
\newcommand\tx{\mathsf{tx}\xspace}
\newcommand\aux{\mathsf{aux}\xspace}
\newcommand\txtype{\mathsf{txtype}\xspace}
\newcommand\swap{\mathsf{swap}\xspace}
\newcommand\mint{\mathsf{mint}\xspace}
\newcommand\collect{\mathsf{collect}\xspace}
\newcommand\burn{\mathsf{burn}\xspace}
\newcommand\verifyTransaction{\mathsf{VerifyTx}\xspace}
\newcommand\verifyBlock{\mathsf{VerifyBlock}\xspace}
\newcommand\btype{\mathsf{btype}\xspace}
\newcommand\meta{\mathsf{meta}\xspace}
\newcommand\inn{\mathsf{in}\xspace}
\newcommand\outt{\mathsf{out}\xspace}
\newcommand\summary{\mathsf{summary}\xspace}
\newcommand\block{\mathsf{B}\xspace}
\newcommand\updateState{\mathsf{UpdateState}\xspace}
\newcommand\leader{\mathsf{leader}\xspace}
\newcommand\elect{\mathsf{Elect}\xspace}
\newcommand\com{C\xspace}
\newcommand\prune{\mathsf{Prune}\xspace}

\newcommand\tokenbank{\mathsf{TokenBank}\xspace}

\newcommand\poolsets{\mathsf{PoolSets}\xspace}
\newcommand\deposits{\mathsf{Deposits}\xspace}
\newcommand\deposit{\mathsf{Deposit}\xspace}
\newcommand\payouts{\mathsf{Payouts}\xspace}
\newcommand\uid{\mathsf{userId}\xspace}
\newcommand\positionId{\mathsf{posId}\xspace}
\newcommand\fees{\mathsf{fees}\xspace}
\newcommand\createpool{\mathsf{createPool}\xspace}
\newcommand\sync{\mathsf{Sync}\xspace}
\newcommand\type{\mathsf{type}\xspace}
\newcommand\amnt{\mathsf{amnt}\xspace}

\newcommand\lp{\mathcal{LP}\xspace}
\newcommand\led{\mathcal{L}\xspace}

\newcommand\summ{\mathsf{sum}\xspace}
\newcommand\SwapTx{\mathsf{Swap}\xspace}
\newcommand\MintTx{\mathsf{Mint}\xspace}
\newcommand\BurnTx{\mathsf{Burn}\xspace}
\newcommand\CollectTx{\mathsf{Collect}\xspace}
\newcommand\FlashTx{\mathsf{Flash}\xspace}

\newcommand\ExIn{\mathsf{ExactInput}\xspace}
\newcommand\ExOut{\mathsf{ExactOutput}\xspace}
\newcommand\SwapCallback{\mathsf{SwapCallback}\xspace}
\newcommand\addLiquidity{\mathsf{addLiquidity}\xspace}
\newcommand\decreaseLiquidity{\mathsf{decreaseLiquidity}\xspace}
\newcommand\MintCallback{\mathsf{MintCallback}\xspace}
\newcommand\PoolDep{\mathsf{PoolDeployer}\xspace}
\newcommand\PoolFac{\mathsf{PoolFactory}\xspace}
\newcommand\NFPM{\mathsf{NonfungiblePositionManager}\xspace}
\newcommand\NFTPD{\mathsf{NonfungibleTokenPositionDescriptor}\xspace}
\newcommand\SwapRouter{\mathsf{SwapRouter}\xspace}
\newcommand\positions{\mathsf{Positions}\xspace}
\newcommand\flash{\mathsf{Flash}\xspace}

\newcommand\amount{\mathsf{amnt}\xspace}
\newcommand\range{\mathsf{priceRange}\xspace}
\newcommand\ltick{\mathsf{lowerTick}\xspace}
\newcommand\utick{\mathsf{upperTick}\xspace}

\newcommand{\ifcond}[1]{\textbf{if} {#1} \textbf{then}}
\newcommand{\forloop}[1]{\textbf{for} {#1} \textbf{do}}

\newcommand{\elsecond}{\textbf{else}}
\newcommand{\pluseq}{\mathrel{+}=}
\newcommand{\mineq}{\mathrel{-}=}

\newtheorem{remark}{Remark}
\newtheorem{theorem}{Theorem}
\newtheorem{lemma}{Lemma}

\def\BibTeX{{\rm B\kern-.05em{\sc i\kern-.025em b}\kern-.08em
    T\kern-.1667em\lower.7ex\hbox{E}\kern-.125emX}}
\begin{document}

\title{ammBoost: State Growth Control for AMMs}
\author{\IEEEauthorblockN{Nicolas Michel}
\IEEEauthorblockA{\textit{University of Connecticut} \\
Storrs, CT, USA \\
nicolas.michel@uconn.edu}
\and
\IEEEauthorblockN{Mohamed E. Najd}
\IEEEauthorblockA{\textit{University of Connecticut} \\
Storrs, CT, USA \\
menajd@uconn.edu}
\and
\IEEEauthorblockN{Ghada Almashaqbeh}
\IEEEauthorblockA{\textit{University of Connecticut} \\
Storrs, CT, USA \\
ghada@uconn.edu}}

\maketitle
\pagestyle{plain}

\begin{abstract}
Automated market makers (AMMs) are a prime example of Web 3.0 applications. Their popularity and high trading activity led to serious scalability issues in terms of throughput and state size. In this paper, we address these challenges by utilizing a new sidechain architecture, building a system called \sysname. \sysname reduces the amount of on-chain transactions, boosts throughput, and supports blockchain pruning. We devise several techniques to enable layer 2 processing for AMMs, including a \emph{functionality-split} and \emph{layer 2 traffic summarization} paradigm, an \emph{epoch-based deposit mechanism}, and \emph{pool snapshot-based and delayed token-payout trading}. We also build a proof-of-concept for a Uniswap-inspired use case to empirically evaluate performance. Our experiments show that \sysname decreases the gas cost by 96.05\% and the chain growth by at least 93.42\%, and that it can support up to 500x of the daily traffic volume of Uniswap. We also compare \sysname to an Optimism-inspired solution showing a 99.94\% reduction in transaction finality.
\end{abstract}


\section{Introduction}
\label{intro}
Cryptocurrencies and blockchain technology provide an innovative model that ignited a movement towards the next generation of the Internet---Web 3.0. Decentralized Finance (DeFi) is a large category under Web 3.0 in which blockchains are utilized to transform traditional financial services, which are usually centrally managed, into fully decentralized ones. Many of these systems operate in an open-access model, thus removing entrance barriers for users, and enabling a transparent and intermediary-free interaction.

Automated market makers (AMMs) are considered a prime example of DeFi services~\cite{xu2023sok}. They build a platform for automated token trading by establishing liquidity pools for token pairs. An AMM is implemented as a decentralized application (dApp) composed of a set of smart contracts deployed on a smart contract-enabled blockchain---where Ethereum is the dominant choice so far. It supports operations to handle trading and liquidity management, such as swaps, mints, burns, and collects. AMMs are a huge industry with a total monthly trading volume of \$46$-$\$95 billion (during the first half of 2023), and an estimated total market cap of \$16 billion as of January 2024~\cite{forbesAMM}. Many popular AMMs are deployed in practice and widely used. Examples include Uniswap~\cite{uniswap}, Curve~\cite{curve}, DODO~\cite{dodo}, and Sushiswap~\cite{sushiswap}, which command a large portion of the AMM market share~\cite{amm-market-share}.

\textbf{Challenges.} At the same time, AMMs are a huge scalability problem. Their popularity and high trading activity led to serious efficiency problems since AMMs produce a massive number of on-chain transactions. On the one hand, this increases the underlying blockchain size, and on the second hand, it incurs large (gas) fees. This large workload does not only amplify state storage cost,\footnote{A dApp on-chain storage includes the state of the dApp's smart contract, \emph{and} the transaction history recorded on the blockchain that produced this state. A misconception about on-chain storage is counting only the contract \emph{latest} state while ignoring the permanent on-chain storage of these transactions.} but also transaction processing/confirmation delays due to the low throughput of blockchains (on Ethereum, average throughput is around 12 transaction/sec~\cite{eth-throughput}, and 62\% and 23\% of all transactions have a delay of at least 5 min and 30 min, respectively~\cite{lao2024unpacking}).  

Concretely, based on the traffic analysis that we conducted for 2023 (see Appendix~\ref{appdx:uniswap-traffic-analysis}), Uniswap V3 produced 20 million transactions on Ethereum in 2023 (transaction sizes range from $400$ bytes to $3000$ bytes). This translates to adding around $20.2$ GB to the Ethereum blockchain. In its year of deployment (Nov 2018), Uniswap V1 generated 34,000 transactions, and in 2023, Uniswap's various versions on Ethereum generated around 80 million transactions, leading to 2,353\% increase in transaction volume. These numbers indicate that the scalability problem of AMMs is amplified over the years. This does not only impact the AMM itself, but also other dApps deployed on the underlying blockchain. Such contention drives users to put in high transaction fees so that miners would prioritize their transactions.

Thus, the challenge is how to address the scalability problem of AMMs \emph{without impacting security} (e.g., without employing trusted third parties or weakening  consensus security).


\textbf{Limitations of prior work.} Existing blockchain scalability solutions target layer 1, such as sharding~\cite{luu2016secure,Danezis16,zamani2018,Kokoris18}, or layer 2, such as rollups~\cite{optimism,zksync}, that enable off-chain processing. Applying these solutions to AMMs impacts performance and security, and does not even cut storage costs. 

In sharding, localized workload division is used to reduce cross-shard transactions. For smart contract-enabled blockchains, this means that a dApp (so the whole AMM) will reside on one shard~\cite{albassam2018chainspace,tao2020sharding}, thus parallel processing among shards is not utilized. Others~\cite{pirlea2021practical} use static analysis to shard dApps by splitting them into commuting functional units that can be executed in any order, and assigning each unit to a shard. This cannot work for AMMs since their (per pool) operation is sequential. Distributing liquidity pools among shards~\cite{pourpouneh2023automated} relies on locked cross-shard transactions (to support multi-swaps and arbitraging), which may degrade the AMM performance. 

\begin{table*}[t] 
\centering
    \resizebox{\textwidth}{!}{
\begin{threeparttable}
    \caption{Comparison between \sysname and rollup solutions.}
    \label{tab:comparison}

    \begin{tabular}{|l | c | c | c| c | c | c| c|}
       \hline
       Solution\tnote{*}  & Type &  Throughput & Token Payout Delay &  Liquidity Withdrawal Overhead & Decentralized & Mainchain Storage \\
       \hline
       Uniswap Optimism~\cite{uniswapUniswapOptimism, optimismWithdrawalFlow} & Optimistic Rollup & 0.6 tx/s & 7 days & 4 tx (including Burn)  & No & Batch-txn transcript \\
        \hline
       Unichain~\cite{unichain, unichainSubmittingTransactions}  &  Optimistic Rollup &   1.92 tx/s & 7 days & 4 tx (including Burn) & Yes  & Batch-txn transcript  \\
        \hline
       ZKSwap~\cite{zkswap, zknationZIP4Reduce, zksyncWelcomeDocs, zksyncBridgingZKsync}   & ZK-rollup & 8 - 25 tx/s & 3-24 hrs & 2-3 tx (including Burn) & No & State changes \\
       \hline
       \sysname & Sidechain & 138.06 tx/s & 346.49 s & 1 (Burn) tx & Yes & State changes\\
       \hline
    \end{tabular}
    \begin{tablenotes}
    \item[*] \footnotesize{For Unichain, average throughput reported for its deployment~\cite{uniscan}. ZKSwap is deployed on top of ZKSync-Era~\cite{zksync}, so reported metrics are for the latter~\cite{zksync-era-finality, zksyncWelcomeDocs}.} 
    \end{tablenotes}

\end{threeparttable}
    } 
\end{table*}

Layer 2 solutions that allow computations, i.e., beyond just currency transfer as in payment networks~\cite{Decker15}, also have limitations. Optimistic rollups have long contestation periods that may reach one week as in Optimism and Arbitrum~\cite{optimism,arbitrum}. Thus, a user has to wait until the end of this period to ensure that submitted results are valid, and withdraw tokens accrued from trades~\cite{uniswapUniswapOptimism,amm-rollup-two}. Moreover, they have security issues; verifiers (who validate submitted results during contestation) could be centralized trusted entities~\cite{rosca2023security}, while decentralized verifier networks have incentive compatibility issues, i.e., configuring a proper incentive model that encourage verifies to vet the results is still an open question~\cite{li2023security,mediumCheaterChecking,mediumOptimisticRollup}, which may lead to adopting incorrect ledger state changes. Also, optimistic rollups are vulnerable to fraud proof denial attacks causing dispute transactions to be delayed, and then rejected, as the contestation period is over~\cite{koegl2023attacks}. 

Zero-knowledge (ZK)-rollups~\cite{starkware,Bowe20,bonneau2020coda} are costly; proof generation is computationally heavy (a prover in zksync era can cost 1365 USD per month~\cite{chaliasos2024analyzing}) and it becomes worse when attesting to complex transactions, and many require a trusted setup. They also may have a long transaction confirmation delay that may reach up to 24 hours as in zkSync Era~\cite{zksync} (with a recent upgrade cutting that to 3 hrs~\cite{zknationZIP4Reduce}). Such delays force users to wait longer for their transactions to be finalized, which negatively impact the operation of AMMs. Some solutions employ a centralized settlement party for trade matching with ZK proofs to prove settlement correctness~\cite{khalil2019tex}, so they are a form of centralized ZK-rollups. 

Another instance of layer 2 solutions is sidechains~\cite{Back14,Garoffolo20,Gavzi19,Kiayias19,Poon17,cosmos,Wood16,btcrelay,rootstock,xdai}. Existing efforts mostly focus on two-way peg, i.e., currency transfer between the chains, and only consider independent sidechains. That is, each chain has its own transactions, miners, and tokens. This narrow focus limits the performance gains that sidechains can achieve, and prevents workload sharing between the chains. Cosmos~\cite{cosmos} allows some form of data exchange using data duplicates, thus amplifying storage costs. For AMMs, Polygon~\cite{polygon} operates an EVM-compatible sidechain hosting the whole AMM, so it is separated from the mainchain, where tokens can be transferred to Ethereum using bridges. Isolating an AMM on a separate blockchain impacts composability with other dApps, and limits interaction with the mainchain to merely currency transfers. Also, it just moves the on-chain storage cost to the sidechain, which will have scalability issues on its own as the AMM user population grows. 

In terms of blockchain pruning, pruning for sharding~\cite{Kokoris18,tian2024slchain} targets the UTXO model; it is not for the account model underlying smart contract-enabled blockchains. Also,~\cite{Kokoris18} moves the responsibility of storing transactions (to prove their existence) to the clients. Snapshot-based pruning solution~\cite{palm2018selective,song2022block,reddy2021secureprune,ethereumPruningGoethereum} either target the UTXO model~\cite{reddy2021secureprune,song2022block}, support only permissioned blockchains~\cite{palm2018selective}, or require archival nodes that store the full blockchain as in Bitcoin and Ethereum~\cite{bitcoin-file-pruning,bitcoin-archival-nodes,ethereumPruningGoethereum}. Optimistic rollups do not prune since the processed transaction batch is stored on-chain~\cite{optimismDerivationStack,arbitrumDataAvailability}, or on a third-party data availability layer~\cite{celestiaIntroductionArbitrum, astriaDataAvailability}, to allow verifiers to verify the submitted results~\cite{chaliasos2024analyzing, optimismDerivationStack}. ZK-rollups either store the batch on-chain~\cite{polygonArchitecturePolygon} or avoid that at the expense of more expensive ZK proofs~\cite{chaliasos2024analyzing, zksyncDataAvailability}. For sidechains, none of the existing solutions support pruning.

We present a comparison with different layer 2 scalability solutions in Table~\ref{tab:comparison}.

\textbf{A new approach.} Existing solutions around sidechains, due to their narrow focus, fall short in realizing the full potential of sidechains in building an effective layer 2 scalability solution. This has been observed in~\cite{chainboost-paper}, who proposed a framework called chainBoost with a new sidechain architecture that has a mutual-dependence relation with the mainchain, thus permitting workload sharing, arbitrary data exchange, and blockchain pruning. chainBoost targets resource markets---Web 3.0 systems that offer decentralized digital services~\cite{filecoin,livepeer}. It directs all heavy/frequent service-related traffic to the sidechain, which in turn processes this traffic and produces summaries of the state changes that are used to sync the mainchain. Once these summaries are confirmed on the mainchain, the temporary blocks containing the actual transactions on the sidechain are pruned. The empirical results in~\cite{chainboost-paper} show substantial performance gains in terms of blockchain size, confirmation delays, and throughput. All of these are achieved without compromising security and while keeping the mainchain as the single truth of the system state.

These advantages motivated us to ask \emph{whether we can utilize dependent-sidechains to control the state growth of AMMs, and boost their throughput, in a secure and low-overhead way without isolating the AMM on a separate blockchain.}

\subsection{Our Contributions} 
We answer this question in the affirmative and propose \sysname; a secure storage control and throughput boosting solution for AMMs. We make the following contributions. 

\textbf{System design.} \sysname introduces a novel approach for dividing the AMM functionality into two modules: one that resides on the mainchain and another that is operated by the sidechain. In particular, \sysname offloads processing most transactions (swaps, mints, collects, and burns) to the sidechain, and minimizes the functionality remaining on the mainchain. The latter is encapsulated in a base smart contract, called $\tokenbank$, that manages the actual tokens by tracking only the transaction summaries produced by the sidechain. It also includes all operations that must happen in real-time on the mainchain, such as flash loans.

\sysname solves several challenges related to combining dependent sidechains with AMMs. The chainBoost framework assumes a mutual-dependency relation between the main and side chains; both are aware of each other and operate in the same domain, i.e., same transaction format, protocol, and miners. This is not the case for \sysname; the AMM is merely a dApp deployed at the application layer, so applying chainBoost as is requires modifying the mainchain protocol which leads to hard forks. However, \sysname's sidechain is impacted by the mainchain since the tokens and the AMM state are on the mainchain. This means that \sysname exhibits a unidirectional dependency relation. Furthermore, the sidechain does not hold custody of actual tokens, these reside on the mainchain, but it has to process transactions correctly and accept only these for which issuing users own tokens on the mainchain. Moreover, since the mainchain miners in \sysname do not track the sidechain, sync-transaction authentication is needed so that only syncs issued by the rightful epoch committee in \sysname are accepted by $\tokenbank$.

We resolve these challenges by introducing several techniques. First, we require the AMM to have its \emph{own miners} to maintain the sidechain. So, like any secure blockchain, these miners have a mining power (with honest majority), and they will be rewarded for maintaining the sidechain using, e.g., AMM's native token. This allows enjoying the scalability benefits of \sysname without introducing core changes to the mainchain protocol that lead to hard forks. Second, to ensure that users owns actual tokens for trading, we introduce \emph{epoch-based deposits}, where a user deposits on the mainchain the anticipated amount of tokens needed to back up her issued transactions during an epoch on the sidechain. 

Third, to allows trading on the sidechain without token custody, we introduce \emph{pool snapshot-based and delayed token-payout trading}. That is, the pool token balances are retrieved from the mainchain at the beginning of the epoch, which are used to compute trade prices processed on the sidechain. These balances evolve during the epoch based on the transactions processed. Users can directly use newly accrued tokens in trading since all balances are tracked, but they cannot withdraw actual tokens until the epoch ends. This is because the sidechain does not hold actual tokens; token payouts and deposit leftover refunds happen when $\tokenbank$ is synced. Fourth, to ensure that only legit sync-transactions accepted on the mainchain, we introduce a \emph{sync authentication mechanism} that combines quorum certificates and threshold signatures ensuring that only syncs issued by the rightful sidechain committees are accepted by $\tokenbank$. 

A large advantage of \sysname's approach to sidechains is maintaining \emph{dApp composability}. That is, the base smart contract on the mainchain maintains the application state. Thus, interactions with other dApps happen on the mainchain, e.g., as in the original Ethereum setting.

\textbf{Security analysis.} We analyze the security of \sysname showing that it is a secure scalability solution as it preserves the correct and secure operation of the underlying AMM. 

\textbf{Implementation and evaluation.} We also build a proof-of-concept implementation for a Uniswap-inspired use case, and conduct experiments to empirically evaluate the performance gains, in terms of blockchain size, confirmation delay, and throughput, that \sysname can achieve. Our experiments show that \sysname achieves a 96.05\% gas cost reduction and 93.42\% chain growth reduction (when compared to a Uniswap version deployed on the Sepolia testnet). Our experiments demonstrate that \sysname can support large traffic volumes, on the order of up to 500x of Uniswap's daily transaction volume. We also study impact of various configuration parameters on the performance gains of \sysname, and compare it to an Optimism-inspired solution showing a 99.94\% reduction in transaction finality.

Lastly, we believe that \sysname's paradigm can enable more optimizations for AMMs, e.g., integration of privacy-preserving techniques. Also, \sysname could be beneficial for other DeFi applications, and dApps in general, as it introduces a framework for operating application-specific sidechains interacting with smart contract-enabled blockchains. We leave exploring such directions as part of our future work.

\section{Background and Challenges}
\label{sec:background}
We provide an overview of AMMs and the chainBoost framework that we use in the design of \sysname.\medskip

\noindent\textbf{Automated market makers (AMMs).}
AMMs build platforms for token trading powered by the users themselves. This is done by establishing liquidity pools such that a pool trades a pair of tokens, say tokens $A$ and $B$. Users are divided into: clients---token sellers and buyers, and liquidity providers (LPs). Providing liquidity comes from the sellers since buying token $A$ requires paying the price in token $B$ (and vice versa), and from LPs who add tokens to the pool and collect fees in return. Constant function market makers (CFMM) is a popular choice for computing the trading price in AMMs. This formula keeps the ratio of token reserves, and consequently prices, in the pool as balanced as possible to reduce price slippage.

At a basic level, an AMM supports several transaction types: for trading, there are (exact input and exact output) swaps, and for liquidity management, there are mints, burns and collects that allow LPs to submit liquidity positions, collect their fees, and withdraw these positions, respectively. A liquidity position is a data structure in the AMM smart contract containing a position ID, the ID (e.g., a public key) of the owner, the amount of liquidity tokens the position owner provided, and the total amount of fees accrued so far. Collects are simple currency transfer transactions in which a position owner requests to withdraw (all or part of) the fee balance in their position.  While burns allow a position owner to withdraw all liquidity balance and fees of their position and deleting the position. AMMs may provide additional services, such as flash loans~\cite{wang2021towards} allowing clients to take advantage of arbitrage opportunities across different platforms. Furthermore, more sophisticated liquidity approaches are being adopted, e.g., concentrated liquidity~\cite{fritsch2021concentrated} that enables defining a price range over which liquidity will be applied to address issues related to inefficient use of provided funds.

The functionality of an AMM is commonly implemented as a set of smart contracts on top of a smart contract-enabled blockchain, where Ethereum is the dominant choice so far. These contracts create and manage the liquidity pools, and provide the API needed to interact with the AMM. AMMs on public blockchains have several financial and security issues~\cite{bartoletti2022maximizing,park2021conceptual,qin2022quantifying}, e.g., front-running attacks, sandwich attacks, miner/maximal extractable value, etc. Understanding and solving these issues are active research areas. We do not discuss these issues further since they are not the focus of this work; we target the storage cost and throughput of AMMs.\medskip

\noindent\textbf{The chainBoost framework.}
chainBoost~\cite{chainboost-paper} is a sidechain-based solution that reduces blockchain size and confirmation delays, and boost transaction throughput. It introduces a new sidechain architecture that has a mutual-dependence relation with the mainchain, which allows sharing the workload between the two chains, and supports pruning stale records. In particular, all service-related transactions that can be summarized go to the sidechain, while the rest stay on the mainchain. The sidechain works in parallel to the mainchain, and operates in epochs and rounds (an epoch is $\omega$ consecutive rounds and a round is the period during which a new block is mined). At the end of each epoch, the mainchain is synced with summaries of the workload processed by the sidechain in that epoch.

The sidechain is managed by the mainchain miners, where for each epoch, a committee is elected to process the sidechain traffic. The rest of the mainchain miners, who are not on the committee, do not process the sidechain traffic, thus reducing their load. To speedup agreement, chainBoost employs a secure practical Byzantine fault tolerance (PBFT)-based consensus (similar to those in~\cite{Kogias16,Gilad17}) for the sidechain.

The sidechain is composed of two types of blocks: temporary meta-blocks and permanent summary-blocks. For each sidechain round, the committee mines a meta-block containing the transactions they processed, so that once a transaction appears in a meta-block it is considered final. In the last round of the epoch, this committee mines a summary-block summarizing all state changes induced by the meta-blocks within that epoch. After that, it issues a sync-transaction containing the summarized state changes, which the mainchain miners use to update the relevant state variables on the mainchain. Once the sync-transaction is confirmed on the mainchain, all meta-blocks used to produce the respective summary-block are discarded. This significantly reduces the sidechain size, and subsequently, the mainchain size. At the same time, having permanent summary-blocks enables anyone to verify the source of the state changes recorded on the mainchain.\medskip

\noindent\textbf{Challenges of using dependent-sidechains for AMMs.} AMMs' setting is different from the one in~\cite{chainboost-paper}: First, in chainBoost the two chains are mutually-dependent and aware of each other. In \sysname, mainchain miners belong to a smart contract-enabled blockchain, and the AMM is simply an application deployed on that blockchain. Modifying the mainchain protocol to have its miners operate the sidechain would lead to a hard fork. To avoid that, the sidechain needs its own miners to run its protocol (indeed a miner can operate on both chains by running two protocols---one for each chain).  

Second, in \sysname, actual tokens reside on the mainchain, while trading activities happen on the sidechain. However, in AMMs, a user who want to place a trade or add liquidity, must prove that they own actual tokens covering the activity by transferring the tokens to the AMM contract before the transaction is processed. Thus, a mechanism is needed to provide guarantees that users own actual tokens to backup transaction they issue on the sidechain. At the same time, transaction execution must be based off the pool balances from the mainchain, and a mechanism is needed to handle actual token payouts on the mainchain resulting from trading activities that took a place on the sidechain.

Second, the mainchain miners in \sysname are not required to track the sidechain, and even they may not be aware of its existence. Nonetheless, $\tokenbank$ must only accept sync-transactions coming from the rightful sidechain epoch committees. Thus, \sysname must implement an authentication mechanism to ensure that only legit syncs are accepted. Devising techniques to address these issues resemble the core novelty of \sysname system design.

\section{Preliminaries}
\label{sec:prelim}
\noindent\textbf{Notation.} We use $\lambda$ to denote the security parameter, $\param$ to denote the system public parameters, $\led$ to denote a ledger (or blockchain), $\led_{\mainc}$ to denote the mainchain ledger, and $\led_{\sidec}$ to denote the sidechain ledger. $\led_{\mainc}$ is the smart contract-enabled blockchain on top of which the AMM base smart contract is deployed, while $\led_{\sidec}$ is the sidechain of the AMM ecosystem. Each party maintains a secret key $\sk$ and a public key $\pk$. Lastly, we use $\ppt$ for probabilistic polynomial time.

\vspace{3pt}
\noindent\textbf{System model.} 
\sysname involves a base smart contract representing the AMM on the mainchain, and a sidechain that processes most of the AMM workload. Anyone can join/leave the AMM at anytime, and these parties are known using their public keys. Participants are three types: clients $\clients$ who only use the AMM trading services, liquidity providers $\lp$ who provide liquidity for the pools operated by the AMM, and miners $\miners$ who maintain the AMM sidechain. We do not place any restrictions on the mainchain beyond being a secure smart contract-enabled blockchain. \sysname operates in rounds and epochs (as defined earlier). The sidechain runs a secure PBFT-based consensus, e.g.,~\cite{Kogias16,Gilad17}, in which a committee is elected for each epoch to run the agreement (we provide a background on PBFT and committee election in Appendix~\ref{apdx:pbft-review}).\footnote{To simplify the presentation, we adopt a leader-based PBFT in which a leader proposes a block for the committee to agree on (as in~\cite{Kogias16}). Nonetheless, voting-based PBFT (as in~\cite{Gilad17}) can be used.} In any epoch, the committee mines new blocks: temporary meta-blocks that record transactions, and permanent summary-blocks that summarize meta-blocks mined in an epoch. The committee also issues sync-transactions to sync the base AMM smart contract deployed on the mainchain. Accordingly, the \sysname framework provides the following functionalities:

\begin{description}
    \item[$\syssetup(1^{\lambda}, \led_{\mainc}) \rightarrow (\param, \led_{\sidec}^0)$:] Takes as input $\lambda$ and $\led_{\mainc}$. It configures the public parameters $\param$, and deploys a base AMM smart contract on $\led_{\mainc}$. It outputs $\param$ and the initial sidechain ledger $\led_{\sidec}^0$ (the genesis block referencing the mainchain block containing the base contract).

    \item[$\partysetup(\param) \rightarrow (\stt)$:] Takes as input $\param$ and outputs the party's initial local state $\stt$, which contains a keypair $(\sk, \pk)$, and in case of miners, the current view of $\led_{\sidec}$.

    \item[$\createTransaction(\txtype, \aux) \rightarrow (\tx)$:] Takes the transaction type $\txtype$ and any additional information $\aux$ as inputs, and outputs a transaction $\tx$ of one of the following types:
    \begin{itemize}
        \item $\tx_{\deposit}$: Allows a user to deposit funds on the mainchain to support their activities on the sidechain.
        \item $\tx_{\swap}$: Allows a client to submit a trade.
        \item $\tx_{\mint}$: Allows an LP to provide liquidity to a pool.
        \item $\tx_{\collect}$: Allows an LP to collect their fees.
        \item $\tx_{\burn}$: Allows an LP to withdraw their liquidity.
        \item $\tx_{\sync}$: Allows a sidechain committee to sync the AMM base contract that resides on the mainchain.
    \end{itemize}

    \item[$\verifyTransaction(\tx) \rightarrow (0/1)$:] Takes as input a transaction $\tx$, and outputs 1 if $\tx$ is valid based on the syntax/semantics of its type, and 0 otherwise.

    \item[$\verifyBlock(\led_{\sidec}, \block_{\btype}) \rightarrow (0/1)$:] Takes as input the sidechain state $\led_{\sidec}$, a new block $\block$ with type $\btype = \meta$ or $\btype = \summary$. It outputs 1 if $\block$ is valid based on the syntax/semantics of the block type, and 0 otherwise.

    \item[$\updateState(\led_{\sidec}, \aux, \btype) \rightarrow (\led_{\sidec}')$:] Takes as input the current sidechain state $\led_{\sidec}$, and a set of pending transactions $\aux = \{\tx_i\}$ (if $\btype = \meta$) or $\bot$ (if $\btype = \summary$ since the inputs are the last epoch meta-blocks from $\led_{\sidec}$). It reflects the changes induced by $\aux$ and outputs a new state $\led_{\sidec}'$.

    \item[$\elect(\led_{\sidec}) \rightarrow (\com, \leader)$:] Takes as input the current state of the sidechain ledger $\led_{\sidec}$, and outputs an epoch committee $\com$ and its leader $\leader$.

    \item[$\prune(\led_{\sidec}) \rightarrow (\led_{\sidec}')$:] Takes as input the current sidechain state $\led_{\sidec}$, and produces an updated state $\led_{\sidec}'$ in which all stale meta-blocks are dropped.
\end{description}

Note that $\updateState$ is the process of mining a new block on the sidechain based on its consensus protocol.   

\vspace{3pt}
\noindent\textbf{Security model.}
We aim to develop a secure state growth control solution that preserves AMM correctness and security. \sysname builds a sidechain, which is a blockchain, that interacts with the application layer of the mainchain through the base AMM smart contract. This sidechain must be a secure ledger that satisfies safety and liveness as defined in~\cite{garay2015bitcoin}.

A secure blockchain must record only valid transactions and blocks, thus its protocol is parameterized by predicates to verify transactions and blocks. For dApps, validity is governed by the code of their smart contracts, and miners ensure that the ledger state changes have been produced by a successful execution of this code. \sysname reduces the AMM functionality deployed on the mainchain, and it processes most of the workload (following the same logic of the AMM) on the sidechain. Thus, in our security analysis, we show that \sysname preserves the security and correct operation (i.e., safety and liveness) of the original AMM. 

\vspace{3pt}
\noindent\textbf{Adversary model.}
We assume a secure mainchain as defined above. For the sidechain, we have honest miners who follow the protocol, and malicious miners controlled by the adversary who may behave arbitrarily. The adversary can deploy new miners or corrupt existing ones, without going above the threshold of faulty nodes required by the sidechain consensus protocol. As we use a PBFT-based consensus for the sidechain, it follows the threshold assumption for PBFT~\cite{Kogias16}; the adversary can corrupt up to $f$ parties, where a committee size is $3f+2$ and $2f+2$ votes are needed to reach an agreement. The adversary can see all messages and transactions sent in the system (since we deal with public permissionless blockchain systems) and can reorder these messages and delay them. As in~\cite{Kokoris18,Kogias16,garay2015bitcoin,pass2017fruitchains,pass2017analysis,abraham2020sync}, we assume bounded-delay message delivery, so any sent message (or transaction) will be delivered within $\Delta$ time. We assume slowly-adaptive adversaries~\cite{avarikioti2019divide} that can corrupt miners only at the beginning of each epoch. Lastly, we deal with $\ppt$ adversaries.

\section{System Design}
\label{sec:design}
\sysname changes the AMM deployment structure, as shown in Figure~\ref{ammboost}. The base smart contract is minimal; mainly tracking token balances of users and liquidity pools, while most transaction processing happen on the sidechain. Summaries of the sidechain processed traffic are used to sync the base smart contract. In this section, we present the design of \sysname including system setup, architecture and operation, handling interruptions, and its security.

\begin{figure*}[t!]
\centerline{
\includegraphics[height= 2.0 in, width = 1.8\columnwidth]{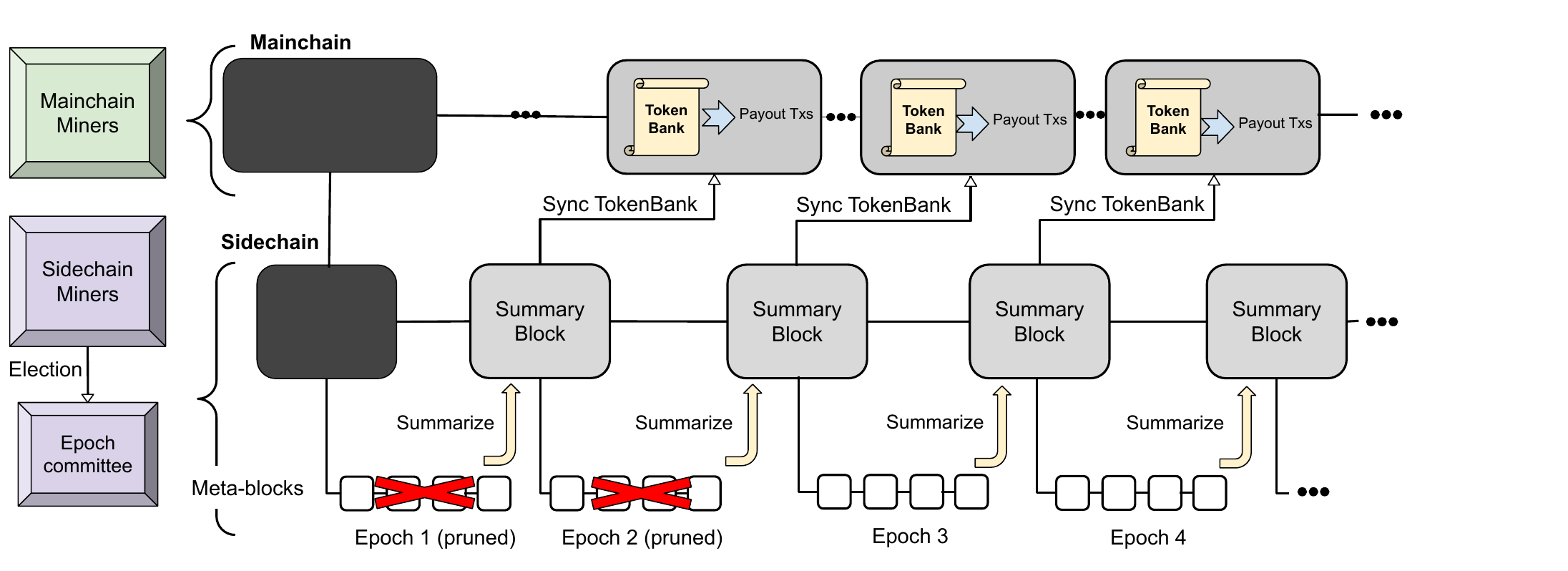}}
\caption{The \sysname framework (Txs stands for transactions).}
\label{ammboost}
\end{figure*}

\subsection{System Setup}
\label{subsec:setup}
The setup phase, see Figure~\ref{fig:syssetup}, mainly specifies the traffic split between the chains and the summary rules for the sidechain traffic, as well as the sidechain parameters such as the epoch length\footnote{The epoch duration impacts syncing frequency. Short epochs mean more sync-transactions, which incurs more gas cost and may impact throughput---as they are processed by the mainchain, however users would receive their tokens faster compared to long epochs. We empirically study the impact of this parameter in Appendix~\ref{appdx:more-eval}.} and its consensus configuration (e.g., committee size). Also, this phase involves deploying the AMM base smart contract on the mainchain and creating the sidechain.\medskip

\begin{figure}[t!]
\begin{framed}
\vspace{-3pt}
\begin{flushleft}
\justifying
\small{
$\syssetup(1^{\lambda}, \led_{\mainc})$: Takes as input $\lambda$ and the current mainchain state $\led_{\mainc}$, and does the following: 

 \vspace{3pt}
 1. Generate the sidechain configuration parameters:
 \begin{itemize}
     \item The epoch length $\omega$.   
     \item All parameters of the sidechain consensus protocol.
     \item Traffic classification rules.
     \item Summary rules and state variables.
 \end{itemize}
 
 2. Deploy $\tokenbank$ on the mainchain.
 
 \vspace{5pt}
\noindent Outputs: epoch length $\omega$, sidechain genesis block $\led_{\sidec}^0$ (that references the block in the updated state $\led_{\mainc}'$ containing $\tokenbank$), and the address of $\tokenbank$.}
\end{flushleft}
\vspace{-10pt}
\end{framed}
\vspace{-6pt}
\caption{System setup.} 
\vspace{-6pt}
\label{fig:syssetup}
\end{figure}

\noindent\textbf{Traffic classification and summary rules.}
AMM transactions and operations are divided into two groups: pool management and trading. Creating and managing token pools, as well as dispensing tokens to clients and LPs, are done on the mainchain since these deal with actual tokens. Flashes are also handled by the mainchain since they require instant token dispensing rather than at the end of the epoch. The rest of the transactions, including swaps, mints, burns, and collects, are handled by the sidechain. In \sysname, the sidechain does not hold custody of tokens, it only tracks their balances based on the processed transactions. Thus, during each epoch, the sidechain produces two structures:
\begin{itemize}
    \item A \emph{payout} list containing users' public keys and the amount/type of tokens they should receive, which is simply the updated deposit balance at the end of an epoch.
    \item A \emph{liquidity position} list containing the position IDs, the public keys of their owners, balances, and any additional information needed by the liquidity management techniques, e.g., price ranges as in concentrated liquidity.
\end{itemize}

The actual token dispensing and deduction happen at the end of an epoch when the sidechain summaries are received. The new state of the pool token balances on the mainchain will be computed based on these lists. Also, the payout encompasses refunding any leftover in the deposits to their owners. In Section~\ref{subsec:operation}, we show the summary rules for each transaction type and how they contribute to the payout and payin lists.\medskip

\noindent\textbf{Base smart contract $\tokenbank$.}
The mainchain part of the AMM is a base smart contract called $\tokenbank$. At an abstract level, as shown in Figure~\ref{fig:tokenbank}, this contract supports creating and managing token pools (i.e., tracking their balances and liquidity positions). It also provides the minimal interface needed to support users' activities on the sidechain, which is mainly creating deposits containing the tokens they want to trade or provide as liquidity. This is needed to guarantee that users own tokens since the sidechain does not hold actual tokens, it only tracks balance evolution. Hence, a user deposits the total amount of tokens they would need during an epoch before this epoch starts, and $\tokenbank$ handles the payouts and payins produced by the sidechain when the epoch ends.

As shown in Figure~\ref{fig:syssetup}, system designers deploy $\tokenbank$ on the mainchain. Once the mainchain block containing this contract is confirmed, the genesis block of the sidechain $\led_{\sidec}^0$ can be created such that it references this mainchain block.\medskip

\begin{figure}[t!]
\begin{framed}
\vspace{-3pt}
\begin{flushleft}
\justifying
\small{
 // ** State variables **
 
\noindent $\poolsets$: token-pair pools managed by the AMM.

 \vspace{3pt}
\noindent $\deposits$: a map of users' public keys and the type/amount of tokens they deposited.

\vspace{3pt}
\noindent $\positions$: a map of users' public keys and the liquidity positions they own.

 \vspace{5pt}
\noindent // ** Functions **
 
\noindent $\createpool(A, B)$: initializes a pool for the token pair ($A$, $B$).

 \vspace{3pt}
\noindent $\deposit(\type, \amnt)$: allows a user to deposit an amount $\amnt$ of token with type $\type$ for the next epoch.

 \vspace{3pt}
\noindent $\sync(\aux)$: Sync the mainchain AMM state with the sidechain epoch summaries. The input $\aux$ contains the updated pool balances and liquidity positions, and the payin/payout lists.

 \vspace{3pt}
\noindent $\flash(\aux)$: Receive a flash loan request where $\aux$ contains all required inputs, then calculate the token amount the pool can provide and initiate the callback process (more details can be found in Section~\ref{subsec:operation}). }
\end{flushleft}
\vspace{-10pt}
\end{framed}
\vspace{-6pt}
\caption{$\tokenbank$ abstract functionality.} 
\vspace{-12pt}
\label{fig:tokenbank}
\end{figure}

\noindent\textbf{Sidechain management.}
Different from chainBoost, the sidechain in \sysname is managed by its own miner population. Like any blockchain, anyone who can join as a miner given that they have mining power (e.g., a stake) to establish Sybil-resistant identities. By doing so, deploying \sysname does not introduce core changes to the mainchain protocol (required if mainchain miners are to manage the sidechain), thus avoiding the consequences of having hard forks.  

At the beginning of each epoch, a committee from these miners is elected (e.g., using random sortition~\cite{Gilad17,Kokoris18}), which runs a PBFT consensus protocol to agree on mining meta/summary blocks and issuing sync-transactions. That is, the committee leader proposes new blocks or sync-transactions, and collect votes from the committee members. Once a vote majority is reached, the new block is added to the sidechain or the sync-transaction is sent to the mainchain. In \sysname, a sync-transaction is basically a call to the function $\sync$ in $\tokenbank$ shown in Figure~\ref{fig:tokenbank}. Executing the $\tokenbank$ contract is like executing any other contract deployed at the application layer of the mainchain.

\sysname does not place any restrictions on the consensus protocol of the underlying mainchain (beyond being a secure one), and the mainchain consensus protocol could be different from the sidechain one. Also, for \sysname's sidechain, any secure PBFT-based consensus and any secure committee election algorithm can be used. Moreover, the committee size has a lower bound that guarantees the honesty threshold required by consensus (formal analysis can be found in~\cite{chainboost-paper}).

\subsection{System Operation---Transaction Processing}
\label{subsec:operation}
\sysname operates in epochs and rounds. The sidechain committee begins the epoch by retrieving the latest state, i.e., pool token balances, liquidity positions, and user deposits from the mainchain. It then processes all valid sidechain transactions, including swap, mint, burn, and collect (so users send these transactions to the sidechain). These are packaged into meta-blocks such that a meta-block is mined in each round. In the last round of the epoch, this committee produces a summary-block capturing the payouts for participating users, and any changes on liquidity positions, where the updated liquidity pool balances will be computed based on these lists. After that, it invokes the $\sync$ function in $\tokenbank$ that resembles submitting a sync-transaction to update the AMM state on the mainchain, which is the state of $\tokenbank$. 

In this section, we describe how the various transactions are processed and summarized (Figure~\ref{fig:summary-rules} captures how the sidechain workload is summarized in \sysname). Before that, we want to point out that although a user will obtain her newly claimed tokens at the end of an epoch, still they can use these tokens immediately during an epoch for trading. This is because new tokens will be added to the user deposit balance, and those that were used are deducted from deposit. Thus, the latest deposit state includes the newly accrued tokens, and it can be used for any trading activity by the user.\medskip

\noindent\textbf{Swaps.}
A swap transaction is a trade between the two tokens managed by a liquidity pool. A client provides an input of tokens and receives an output based on the price derived from the pool token balances and any user-defined trade conditions.

\begin{figure}[t!]
\begin{framed}
\vspace{-3pt}
\begin{flushleft}
\justifying
\small{
Input: meta-blocks $\block^1_\meta, \dots, \block^n_\meta$ from an epoch and $\deposits$ (the latter is the one retrieved from the mainchain at the beginning of the epoch).

\vspace{3pt}
\noindent Initialize: summary structures $\summ_\payouts$ and $\summ_\positions$.

\vspace{3pt}
\noindent \forloop{$i \in \{1, \dots, n\}$ and every $\tx \in \block^i_\meta$}

\ifcond{ $\tx.\txtype = \tx_{\swap}$ }

\;\;\; $\deposits[\tx.\uid].\amount[\inn.\type] \mineq \tx.\amount_{\inn} $

\;\;\; $\deposits[\tx.\uid].\amount[\outt.\type] \pluseq \tx.\amount_{\outt} $

\;\;\; Update fees in $\summ_\positions$ for all positions used to fill $\tx$

\;\;\;\;// Liquidity amounts are computed as explained under 

\;\;\;\;// mints and burns.

\elsecond \ifcond{ $\tx.\txtype = \tx_{\mint}$ }

\;\;\; $\summ_\positions[\tx.\positionId].\amount_A \pluseq \tx.\amount_A $

\;\;\; $\summ_\positions[\tx.\positionId ].\amount_B \pluseq \tx.\amount_B $

\;\;\; $\summ_\positions[\tx.\positionId ].\range = $

\quad\quad\quad\quad\quad\quad\quad\quad\quad\quad $(\tx.\ltick,\tx.\utick) $

\;\;\; $\deposits[\tx.\uid].\amount_A \mineq \tx.\amount_A $

\;\;\; $\deposits[\tx.\uid].\amount_B \pluseq \tx.\amount_B $

\elsecond \ifcond{ $\tx.\txtype = \tx_{\burn}$ }

\;\;\; $\summ_\positions[ \tx.\positionId].\amount_A  \mineq \tx.\amount_A $

\;\;\; $\summ_\positions[\tx.\positionId ].\amount_B \mineq \tx.\amount_B $

\;\;\; $\deposits[\tx.\uid].\amount_A \pluseq \tx.\amount_A $

\;\;\; $\deposits[\tx.\uid].\amount_B \pluseq \tx.\amount_B $

\elsecond \ifcond{ $\tx.\txtype =\tx_{\collect}$ }

\;\;\; $\summ_\positions[\tx.\positionId].\fees_A  \mineq \tx.\amount_A $

\;\;\; $\summ_\positions[\tx.\positionId].\fees_B  \mineq \tx.\amount_B $

\;\;\; $\deposits[\tx.\uid].\amount_A \pluseq \tx.\amount_A $

\;\;\; $\deposits[\tx.\uid].\amount_B \pluseq \tx.\amount_B$

\noindent Output $\summ_\payouts = \deposits$,  and  $\summ_\positions$}
\end{flushleft}
\vspace{-8pt}
\end{framed}
\vspace{-6pt}
\caption{Summary rules ($\uid$ is the user ID and $\positionId$ is the liquidity position ID). Updated liquidity pool balances are computed by TokenBank (as part of processing $\sync$) based on the updated liquidity position and payout lists.} 
\vspace{-6pt}
\label{fig:summary-rules}
\end{figure}

In order to execute a swap transaction, a user's deposit must cover the input token amount. An \emph{exact input swap} transaction contains: the type and amount of input tokens to be traded, the minimum amount of output tokens the trade will accept (as a protection against slippage), a price limit that the trade should not exceed, and a deadline which is a round number after which the trade becomes invalid if not executed by that time. The recipient of the traded tokens is the issuer of the swap (this can be extended to support stating an explicit recipient other than the issuer). For an \emph{exact output swap}, the goal is no longer to trade the exact amount of input tokens for the maximum amount of output tokens, but rather to minimize the amount of input tokens required to trade for the desired exact output. As such, the arguments of the function naturally change to reflect that, with the minimum output slippage protection changing to a maximum input slippage protection. 

\emph{Processing.} This is done using the original AMM logic for price balancing and output calculation. That is, \sysname does not change the logic based on which an AMM operates, it just migrates that to the sidechain (this applies to the rest of the transactions as well). For an exact input swap, the sidechain committee computes the maximum amount of output tokens the user will receive for all of the input tokens provided. While for an exact output swap, the committee computes the minimum amount of input tokens needed to purchase the defined output. In both cases, these computations are based off the updated pool balance on the sidechain. In other words, as transactions are processed, the committee updates the pool state that was retrieved at the beginning of the epoch.

Furthermore, the fees for LPs whose liquidity was used in filling a swap will be computed. To elaborate, when a user submits a swap transaction, they pay a small additional fee in the token pair of the pool. These fees are split proportionally among the positions (based on the amount of liquidity they provide) that occupy the price range for which the swap was executed. \sysname maintains a per-position fee balance, which is updated after every swap, again using the same logic used by the underlying AMM to compute these fees.

Lastly, recall that a user will not get her actual traded tokens until the end of the epoch. However, she can use these tokens for trading on the sidechain since the sidechain tracks all balances. So basically, the deposit is a tuple of two values; one for each token type. When a swap is executed, the input token amount is deducted from the user's deposit while the output token amount is added to this deposit, thus allowing the user to use the newly accrued tokens immediately. 

\emph{Summary rules.} The committee summarizes swaps as follows: for every client, all their swaps are combined into a single tuple containing the client public key and the total payout they should receive. The latter encapsulates both a deduction from their deposit and a refund of any deposit leftover. For example, say a client started with a deposit of $(10 A, 15 B)$ and issued one swap during an epoch that traded 5 token $A$ for 10 token $B$. The updated deposit would be $(5 A, 25 B)$, which represents a payout of 10 $B$, a deduction of 5 $A$, and a refund of deposit leftover of $(5 A, 15 B)$.\medskip

\noindent\textbf{Mints.}
Mint transactions allow the creation of new liquidity positions or modifying existing ones. An LP broadcasts a mint transaction to the sidechain that contains: the lower and upper ticks, representing the price range for which the liquidity can be used, and the type/amount of the token to be used as liquidity. The mint will be accepted if the issuer LP's deposit can cover the provided liquidity amount.

\emph{Processing.} This is also processed using the same logic used by the AMM. We employ a simple approach to track ownership of positions; the sidechain committee generates a unique identifier (e.g., the hash of the mint transaction and the LP's public key) for a new position, and the owner is the public key of the issuer LP. An existing position will receive an increase in its balance (or any other modifications on its price range) after verifying that the transaction issuer is indeed the rightful position owner. Mint transactions are initially invoked with a desired amount of token $A$ and token $B$ as input. The underlying AMM algorithms compute the maximum amount of liquidity (based off the input tokens) that the pool can take in at the current moment from both token types. These values represent the share of the pool liquidity now owned by the newly minted or modified position. The committee then deducts the provided liquidity amount (from both token types) from the corresponding LP's deposit balance.

\emph{Summary rules.} All mints are summarized as a list of liquidity positions with each position consisting of a tuple containing: the position identifier, the public key of the owner, the total amount of liquidity provided (or net change) for each token type, and total amount of accrued fees. Note that the payin/payout of the LP is also updated when summarizing mint transactions; all provided liquidity token amounts are deducted from their deposits as shown in Figure~\ref{fig:summary-rules}.\medskip

\noindent\textbf{Burns.}
A burn transaction allows a partial or complete liquidity withdrawal of a position. It is issued by an LP and contains: the position ID, the tick price limits, and the desired amounts of tokens $A$ and $B$ to be burned. 

\emph{Processing.} Processing a burn transaction boils down to determining if the issuer LP owns the position, then calculating the amount of liquidity this LP owns in a share, and converting that amount of liquidity into an amount of both tokens managed by the pool (using the original logic of the AMM). This would lead to updating the position range (upper and lower price ticks), or a deletion of the position if all its associated liquidity is withdrawn. If a deleted position has fees owed to it, the owner LP will receive these fees as part of her total payout computed at the end of the epoch, i.e., will be added to her deposit balance.

\emph{Summary rules.} Burns are summarized as part of summarizing mint transactions detailed above. Burns adjust the net changes of the pool liquidity balance, i.e., they decrement this balance. Any fully withdrawn position will be removed from the $\tokenbank$ state. The withdrawn liquidity will be added to the LP's deposit balance to be reflected on the payout.\medskip

\noindent\textbf{Collects.}
Collects allow LPs to collect the fees earned by their liquidity positions. An LP broadcasts a collect transaction containing the position identifier and fee amount to be collected.

\emph{Processing.} This includes determining if the issuer LP owns the position, and checking if the amount they want to collect can be covered by their fee balance. If all is fine, the issuer LP's deposit is updated to reflect the amount of collected fees, and the fee balance for that position is adjusted accordingly. 

\emph{Summary rules.} Summarizing collect transactions is also part of summarizing mints/burns and the LP payout structure. That is, fee balance of the referenced position is decreased based on the collected amount, and the payout to the issuer LP is computed based on their updated deposit balance (to which the collected fee amount has been added).\medskip

\noindent\textbf{Flashes.} 
Flash transactions allow users to request short-term loans within the duration of one mainchain block. The goal is to exploit arbitrage opportunities across trading platforms (e.g., different AMMs or other exchanges), thus achieving quick profits due to token price discrepancies. As a result, different from swaps/mints/collects, flash loans are not an internal AMM operation, but rather a short-term lending service that some AMMs offer, and it needs a transfer of \emph{actual tokens} to the loanee rather than tracking internal token balances. 

Flash loans are the only transaction type that \sysname does not offload to the sidechain; the delay in paying out the actual tokens (which happens at the end of an epoch) limits the intended use of flash loans that span a very short period. As such, flash transactions happen on the mainchain as in the original AMM architecture. Since flash loans take place in a singular block, they do not impact the pool balances; the loaned tokens must be returned within one block period or the loan will be inverted. As a result, they do not invalidate any of the transactions processed on the sidechain based off the balance snapshot taken at the beginning of an epoch. 

\begin{remark}
In terms of user experience, clients and LPs should be connected to the side and main chains, with their wallets issuing transactions each chain based on the transaction type. Also, receiving actual tokens is delayed until the end of the epoch. Since a user can use these immediately for trades within an epoch, the delayed payout has no impact. However, if a user wants to use these tokens on a different AMM, for example, then they have to wait until the epoch ends. Nonetheless, this delay is still way shorter than operating the AMM solely on the mainchain (i.e., without \sysname) due to congestion/low throughput issues of blockchains as discussed in the introduction and as we show in Section~\ref{sec:eval}.
\end{remark}

\subsection{System Operation---Chain Management}
In this section, we discuss the syncing process, sidechain pruning, and how \sysname recovers from interruptions.\medskip

\noindent\textbf{Syncing $\tokenbank$.}
The sidechain committee leader, after producing a summary-block for an epoch, calls the $\sync$ function in $\tokenbank$ that resembles a sync-transaction submission. The inputs to this function call include: the list of payouts for all clients and LPs, and the list of liquidity positions with their updated information. 

\emph{Authentication.} $\tokenbank$ must ensure that a $\sync$ invocation is issued by the rightful sidechain committee. We use the idea of quorum certificates (QC)~\cite{kwon2014tendermint,yin2019hotstuff} combined with threshold signatures, which we refer to as TSQC. To authenticate the $\sync$ call for epoch $e+1$, the election of committee $e+1$ must happen during epoch $e$. Then, this committee runs a distributed key generation (DKG) protocol~\cite{bacho2022adaptive} to generate a public verification key $\vk_c$ for the committee and secret shares of the signing key (one share per member) with a threshold of $2f+2$. This committee initiates an agreement on $\vk_c$, and then sends the agreement output to committee $e$ along with proofs of election of each member who participated in the agreement.\footnote{In our implementation, this election proof is the output of the verifiable random function (VRF) used in the election mechanism.} Committee $e$ verifies the election proofs, and then verifies that there is an agreement on $\vk_c$. It then records $\vk_c$ on $\tokenbank$ by adding that to the $\sync$ function call inputs they submit at the end of epoch $e$. During epoch $e+1$, committee $e+1$ runs an agreement over the $\sync$ function call inputs and signs using their signing key shares, which result in one signature over these inputs. The leader then invokes $\sync$ with the inputs and this signature. In turn, $\tokenbank$ verifies the signature using the recorded $\vk_c$. By the security of the threshold signature scheme, this signature will be valid only if at least $2f+2$ committee members has signed.

\emph{Processing.} If successfully verified, $\tokenbank$ processes the $\sync$ function call by updating the list of positions based on the summaries, i.e., delete fully withdrawn positions, create new positions, or adjust existing ones, Then, it updates the pool balance based on the reported payouts/updated positions. Lastly, it dispenses the payouts to the referenced clients/LPs.\medskip

\noindent\textbf{Sidechain pruning.} \sysname uses the \textit{block suppression technique} from chainBoost.  Once the transaction encapsulating the $\sync$ call is confirmed on the mainchain, all meta-blocks associated to this transaction will be pruned. Summary-blocks, as mentioned before, are permanent and represent checkpoints for sidechain state in each epoch. So they can be used to verify the state of the AMM reflected by $\tokenbank$ state variables.\medskip

\noindent\textbf{Handling interruptions.}
We identify the scenarios that can lead to operation interruption in \sysname and how to recover from them. Recall that sidechain committees use a secure PBFT-based consensus that assumes up to $f$ of the elected miners can be malicious. This is valid under a committee size that guarantees satisfying this condition with overwhelming probability, where we adopt the committee size lower bound from~\cite{chainboost-paper}. Thus, a committee will not agree on invalid blocks, and the only interruptions a committee might have are due to having a malicious or unresponsive leader. This leader may either propose an invalid meta/summary-blocks or invalid function call to $\sync$, or not initiate the agreement in the first place. Another interruption could result from rollbacks on the mainchain. That is, when the mainchain miners switch their canonical chain to the one satisfying a particular fork resolution criteria (i.e., the longest branch, or the heaviest one), causing the most recent blocks to be abandoned. This is an issue if the abandoned blocks contain $\sync$ transactions.  

Detection and recovery from these interruptions are done as in chainBoost~\cite{chainboost-paper}. Briefly, a leader that proposes an invalid block or $\sync$ call will be easily detected by the committee when verifying this proposal. Once detected, view-change~\cite{castro1999practical} triggers electing a new leader. For an unresponsive leader, if no agreement is initiated within a timeout period, view-change is invoked. For a leader that proposes invalid $\sync$ inputs, view-change will not help since this happens at the epoch end when it is time for the new committee to take over. Hence, this case, and the rollback interruption, are addressed using mass-syncing. The new committee issues a $\sync$ call covering the summaries they produced in their epoch and those produced earlier in the impacted epochs.

\begin{remark}
A sidechain operates as a regular blockchain, thus any transactions that have not been processed in an epoch will be carried over to the epoch after. All sidechain miners receive transactions destined to the sidechain, but only the elected committee mines meta and summary blocks. Thus, when a new meta-block is mined, the committee and all other sidechain miners remove all published transactions from their queues.
\end{remark}

\subsection{Security}
\label{subsec:security}
Since \sysname delegates AMM transaction processing to the sidechain, and introduces pruning and state synchronization, we show that under this new architecture, operation correctness and security (i.e., safety and liveness) of the underlying AMM are preserved. In Appendix~\ref{apdx:sec-analysis}, we prove the following theorem:

\begin{theorem}\label{th:security}
\sysname preserves the safety and liveness of the underlying AMM.  
\end{theorem}

Briefly, \sysname preserves the correct behavior of the AMM since it processes the sidechain workload using the same logic adopted by the AMM itself. Thus, all transaction types---swaps, mints, collects, and burns---will produce the same outcome, or state changes, as if they are processed on the mainchain. Also, since the sidechain adopts a secure PBFT consensus protocol, that satisfies safety and liveness, the committee only agrees on valid meta-blocks that contain transactions and state changes conforming with the AMM operation and rules. Also, they only agree on valid summary-blocks produced based on the summary rules set in the system that correctly reflect these state changes. Furthermore, liveness of this protocol guarantees that valid submitted transaction will be processed/published on this sidechain. The sidechain also satisfies public verfiability since meta-blocks do get pruned until their sync-transaction is confirmed on the mainchain.

For the AMM base smart contract, recall that it is a dApp deployed on top of a secure smart contract-enabled blockchain (i.e., the mainchain), and its safety/liveness is guaranteed by the security of this chain. That is, mainchain miners will process only valid transactions (including sync-transactions authenticated using TSQC) based on $\tokenbank$ code, and any invalid transaction will be rejected. Lastly, we note that \sysname handles any interruptions on the sidechain side and roll-backs on the mainchain, as discussed in the previous section, using the leader change and mass-syncing mechanisms.

\section{Implementation}
\label{sec:impl}
To assess the performance gains of \sysname, we implement a proof-of-concept and conduct various experiments (our code can be found at~\cite{code-repo}). We chose a Uniswap-inspired use case to represent the underlying AMM. This section discusses the implementation, while the next section shows our results.

\textbf{Sidechain implementation.}
We use the code from~\cite{chainboost-repo}, which employs cryptographic sortition-inspired election mechanism~\cite{Gilad17} for committee election, and this committee runs a BLS collective signing (CoSi)-based PBFT-based consensus~\cite{github:cothority}. We add the modifications needed for \sysname; our sidechain has its own miners, and adopts the summary rules from Section~\ref{sec:design}. We also modify the syncing process to be an invocation to the $\sync$ function in $\tokenbank$ authenticated using TSQC for which we use a simplified version of the golang BLS library~\cite{githubGitHubBjornvdLaanBGRVerify}, and a pre-generated key to sign the $\sync$ transaction. Furthermore, our sidechain implements two extra functions to aid in performing the AMM functionality:

\begin{itemize}
    \item $\mathsf{CreateTx}_\sync$: invoked by the sidechain committee at the end of any epoch to create the $\sync$ call inputs based on the summary-block, including payouts, updated liquidity positions, and liquidity pool balance.

    \item $\mathsf{SnapshotBank}$: retrieves users' deposits at the beginning of an epoch. \sysname retrieves pool balances only for newly created pools, their updated balances can be computed by the sidechain based on the traffic it processes.
\end{itemize}

\textbf{Mainchain details.} 
For the mainchain, we utilize the Ethereum Sepolia testnet~\cite{sepolia} using the hardhat development environment~\cite{hardhat}. We implemented $\tokenbank$ in Solidity~\cite{solidity} and deployed it on Sepolia. In our implementation, the interfacing between the sidechain miners and the mainchain is handled through functionality provided by the Go-Ethereum project~\cite{go-eth}. To allow $\tokenbank$ to authenticate the $\sync$ function call, we implement BLS signature verification in solidity, where we use the 256-bit Barreto-Naehrig (BN256) curve operations defined in the Ethereum precompiles~\cite{reitwiessner2017eip,buterin2018eip}. We implement our hash-to-point functionality as the scalar multiplication of a Keccak256 hash of the $\sync$ entries and the generator of the $G_2$ curve of BN256.

\textbf{Use case: Uniswap-inspired AMM}. We implement swaps, mints, burns, and collects using the same logic as in Uniswap (Appendix~\ref{sec:uniswap}). We do not implement flashes since they represent a very small portion of the traffic, and thus will not impact the performance gains we report. For simplicity, our implementation manages a single pool, where we deployed two standard ERC20 contracts to provide the token pair traded in this pool and used in both the ammBoost and baseline experiments. To test Uniswap V3 in an isolated environment, we use the UniswapV3Factory contract to deploy a pool hosting the two ERC20 tokens. To test against the baseline implementation of Uniswap V3, we wrote and deployed a smart contract to interface with the various Uniswap contracts as detailed in the Uniswap documentation~\cite{UniswapPoolInteractionGuide}. This interface contract routes swaps to the swapRouter contract and mints/burns/collects to the Nonfungible position manager (NFPM) contract. Additionally, it manages all liquidity positions (through the ERC721Reciever interface) created by the users in our experiments.

\textbf{Traffic generation.} 
Users generate traffic on both the mainchain and the sidechain. On the sidechain, the traffic follows the same distribution as in Uniswap (see Appendix~\ref{appdx:uniswap-traffic-analysis}), i.e., 93.19\% of the traffic is $\tx_\swap$, 2.14\% is $\tx_\mint$, 2.38\% is $\tx_\burn$ and 2.27\% is $\tx_\collect$. Our implementation provides configuration settings to modify the distribution and volume of the generated transactions to test their impact on the reported performance metrics.

\section{Performance Evaluation}
\label{sec:eval}

\subsection{Experiment Setup}
We deploy our system on a computing cluster composed of 8 hypervisors, each running a 12-Core, 130 GiB RAM, VM, connected with 1 Gbps network link. This setup is capable of running around 8000 sidechain miners. Unless stated otherwise, an experiment length is 11 epochs, each of which consists of 30 sidechain rounds (a round lasts 7 sec).\footnote{This is a reasonable round duration in light of recent systems, e.g., in Polygon~\cite{polygon-block} and Algorand~\cite{algorand-block} a round lasts for 2 and 2.8 sec, respectively.} Our default meta-block size is 1 MB and a sidechain committee contains 500 miners---a comparable size to those in~\cite{zamani2018, Kokoris18, luu2016secure} (we test the impact of varying the committee size in Appendix~\ref{appdx:more-eval}). We deploy 100 AMM users, generating traffic that arrives at a constant rate of $\rho = \lceil \frac{V_D \times b_t}{3600 \times 24} \rceil$ where $V_D = 25 \times 10^6$ is the chosen daily volume of transactions. We measure the following metrics: 
\begin{enumerate}
    \item \textit{Throughput}: Number of transactions processed per sec.
    \item \textit{Sidechain transaction latency}: The delay between a transaction submission and its appearance in a meta-block. For accurate latency representation, and thus processing a comparable traffic amount, we empty the transaction queues after the end of each run.
    \item \textit{Mainchain transaction latency}: The delay between a transaction submission and its confirmation on Sepolia.
    \item \textit{Payout latency}: The delay between a transaction submission and the completion of $\tokenbank$ syncing in the epoch in which this transaction has been published on the sidechain. We measure this metric by reporting the sum of the sidechain transaction latency, the time needed to issue the $\sync$ call, and the time needed to process the transaction encapsulating the $\sync$ call on the mainchain.
    \item \textit{Gas cost}: The average number of gas units paid to process core transactions. 
    \item \textit{Main and side chain growth}: The growth (in bytes) of both the main and side chains. 
\end{enumerate}

\subsection{Comparison with the Baseline}
We compare \sysname against a baseline, which is a deployment of Uniswap V3 on Sepolia, as mentioned earlier.

\vspace{3pt}
\noindent\textbf{On-chain (itemized) per-operation overhead.} We evaluate the overhead of the deposit and the syncing processes in \sysname, and compare that to the baseline Uniswap on-chain operations. We set the daily volume $V_D$ to be $500$K transactions (10x Uniswap). We use a Gas Profiler~\cite{tenderly} to measure the gas cost of the different components of the $\sync$ transaction. As shown in Table~\ref{tab:gasLatencyBenchmark-ammboost}, we find that storing the state of the liquidity positions is the most expensive, as each consists of 192 bytes (or 6 words), incurring 22,100 gas units per word. The same gas cost of 22,100 gas unit per word is incurred when storing the liquidity pool balance. Each payout transaction incurs a constant fee of 15,771 units. A deposit of two tokens incurs a total cost of 105,392 units. The threshold signature-based quorum certificate incurs a fixed fee that corresponds to the gas cost of the BN256 operations needed for verification, and a fee that is proportional to the length of the summary data structure.

\begin{table}[t!]
  \caption{Mainchain latency and itemized gas cost for \sysname operations ($|sum|$ is the size of the summaries).}
    \label{tab:gasLatencyBenchmark-ammboost}
    \centering
         \fontsize{12pt}{14pt}\selectfont
    \resizebox{\linewidth}{!}{
    \begin{NiceTabular}{|l|c|c|c|c|c|c| }
    \hline
     \textbf{Operation} &   \multicolumn{5}{c|}{\textbf{Sync}} &   \textbf{De-} \\
     \cline{1-6}
     \textbf{Module} &  Payout & Storage & \multicolumn{3}{|c|}{Authentication} &  \textbf{posit}\\
    \cline{4-6}
     & (each) & (per&\multicolumn{2}{|c|} {Hash To Point} & \multicolumn{1}{c|}{Verify} & (2 tokens)\\
    \cline{4-5}
     & & 32 byte)& Keccak256 & ecMUL & \multicolumn{1}{c|}{Pairing} & \\
     \hline
     \textbf{Avg. gas}  &  15,771 & 22,100 & $30 + 6 \times \lceil \frac{|sum
     |} {256} \rceil$  & 6,000 & \multicolumn{1}{c|}{113,000} & 105,392\\
    \hline
    \textbf{MC. lat. (s)} &   \multicolumn{5}{c|}{15.28} &   54.60 \\
    \hline
\end{NiceTabular}
}
\end{table}

\begin{table}[t!]
    \caption{Mainchain latency and gas cost for Uniswap.}
    \label{tab:gasLatencyBenchmark-uniswap}
        \resizebox{\linewidth}{!}{
          \fontsize{6pt}{8pt}\selectfont
    \centering
    \begin{tabular}{|l|c|c|c|c|}
    \hline
       \textbf{Operation}   &  \textbf{Swap} &  \textbf{Mint} & \textbf{Burn} & \textbf{Collect}\\
       \hline
        \textbf{Avg. gas}  & 160,601.45 & 435,609.86 & 158,473.43 & 163,743.04\\ 
        \hline
        \textbf{\textbf{MC. lat. (s)}} & 31.34 & 42.24 & 12.72 & 13.45 \\
        \hline
    \end{tabular}
    }
\end{table}

Overall, the gas cost of the $\sync$ call is affected by the number of positions processed in an epoch and the updated pool balance; this cost does not scale with the number of processed transactions, but rather with the number of clients and liquidity providers. On the other hand, in baseline Uniswap the gas cost is proportional to the total generated traffic, where the numbers in Table~\ref{tab:gasLatencyBenchmark-uniswap} are per one transaction from each type. For the average latency, a $\sync$ transaction does not depend on any other mainchain transactions, so it is confirmed within one block on average. However, since a two-token deposit depends on two ERC20 approvals, and performs 2 transfers, it takes around 4 blocks in our experiments. The same behavior is observed in the Uniswap baseline, as a swap requires one approval from the user, it takes a minimum of two rounds to be processed, and a mint requires two approvals, thus it takes three rounds at least, if the operations are done sequentially.

We also report the per-operation storage cost. In particular, we report the cost breakdown for the $\sync$ call on the mainchain and the summary-block size for \sysname, and the transaction sizes for baseline Uniswap on the mainchain. For the $\sync$ call, the sizes of payout and position entries vary greatly between the summarized changes in a summary-block and the $\sync$ inputs submitted to the mainchain. This is due to the difference in encoding and binary packing between the side and main chains. On the mainchain, Ethereum's application binary interface (ABI) packing keeps track of the data and all the information needed to reinterpret it back, while on the sidechain we use simple binary packing. We also have an extra 192 bytes storage overhead on the mainchain needed for the BLS signature and its public key (namely, $\vk_c$) for authenticating the $\sync$ call. Our findings are in Table~\ref{tab:sizes}.

For Uniswap, we notice that the transactions on Sepolia are smaller than the ones we observe on Ethereum (Appendix~\ref{appdx:uniswap-traffic-analysis}). This is because these chains use different Uniswap transaction routers. The calls to the universal router used on Ethereum end up requiring more arguments, resulting in longer transactions. Uniswap on Sepolia deploys a simpler transaction router and without the universal router (while Uniswap V3 on Ethereum deploys both contracts). 

In general, Uniswap incurs a larger storage cost as its transactions are quite large and all are logged on the mainchain. While for \sysname, only (the less frequent) $\sync$ call transaction is logged on the mainchain.

\begin{table}[t!]

    \caption{Operation storage overhead.}
    \label{tab:sizes}
    \resizebox{\linewidth}{!}{
    \centering
      \fontsize{7pt}{9pt}\selectfont

    \begin{tabular}{|l|c|c|c|c|}
    \hline
      \textbf{\sysname}   &  \textbf{Payout} & \textbf{Position} & \multirow{2}{*}{\textbf{$\vk_c$}} & \multirow{2}{*}{\textbf{Signature}} \\
    \textbf{$\sync$ component}   & \textbf{entry}  & \textbf{entry} &  &  \\
      \hline
       \textbf{Size on Mainchain (B)}  & 352 & 416 & 128 & 64   \\
       \hline
       \textbf{Size on Sidechain (B)}  & 97 & 215 & \multicolumn{2}{|c|}{\cellcolor{gray}} \\
       \hline
       \hline
       \textbf{Uniswap operation} & \textbf{Swap} & \textbf{Mint} & \textbf{Burn} & \textbf{Collect} \\ 
        \hline
       \textbf{Size on Mainchain (B)} & 365.27 & 565.55 & 280.21 & 150.18 \\
    \hline
    \end{tabular}
}
\end{table}

\begin{figure}[t!]
\begin{subfigure}{.5\columnwidth}
  \centering
  \includegraphics[width=\linewidth]{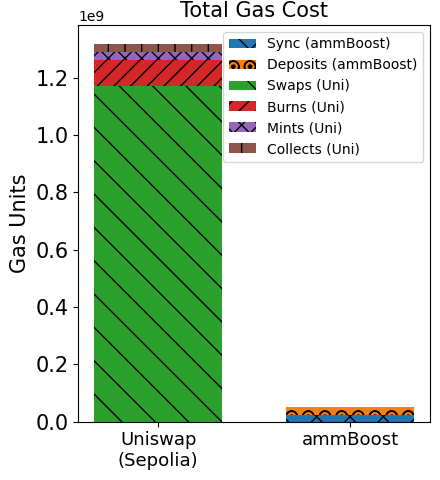}

\end{subfigure}%
\begin{subfigure}{0.5\columnwidth}
  \centering
  \includegraphics[width=\linewidth]{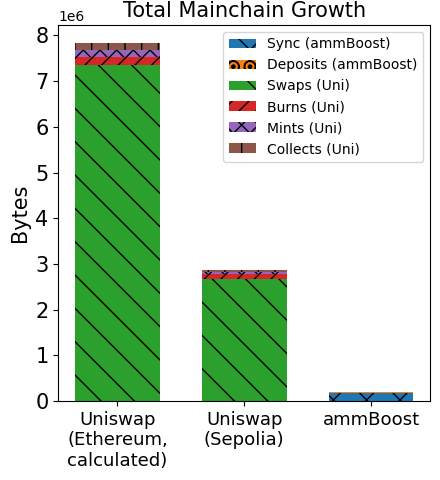}

\end{subfigure}
\vspace{-8pt}
\caption{Gas cost and chain growth comparison.}
\label{fig:test}
\end{figure}

\vspace{3pt}
\noindent\textbf{Overall comparison.} We report the total gas cost and the mainchain state growth of the baseline Uniswap and \sysname. We set the daily volume $V_D$ to be $500$K transactions (10x Uniswap) with the default traffic distribution. We measure the overall mainchain gas cost of relevant operations, and the state growth of the mainchain. 

As shown in Figure~\ref{fig:test}, even if the sync transactions end up being heavy on gas as the number of positions and payouts increases, we achieve a 96.05\% gas reduction when compared to Uniswap Sepolia. The high gas cost of the $\sync$ transaction is offset by it being uncommon (one occurrence per epoch). On the other hand, the gas cost of swaps, mints, burns, and collects in Uniswap are high since all are processed on the mainchain (while in \sysname these are processed on the sidechain). A similar trend is observed for the mainchain state growth, where \sysname provides 93.42\% decrease in growth compared to Uniswap on Sepolia, and 97.60\% decrease when compared to Uniswap on production Ethereum.\footnote{The growth for Uniswap on production Ethereum is calculated by multiplying the count of each transaction type in our experiment by its size as reported in Appendix~\ref{appdx:uniswap-traffic-analysis}.}

\subsection{Impact of Parameter Configuration}
We study the impact of parameter configuration on \sysname's performance, including scalability, traffic distribution, block size, sidechain round duration, number of rounds per epoch, and committee size. Due to space limitation, we only show the full results for the scalability experiment, while the rest can be found in Appendix~\ref{appdx:more-eval}.

\vspace{3pt}
\noindent\textbf{Scalability.} In this experiment, we test the scalability of \sysname (for a single pool) to understand its behavior under heavy traffic. We follow the same traffic distribution as in Uniswap and vary the daily volume $V_D \in \{50K, 500K, 5M, 25M\}$. We record the impact on throughput and transaction/payout latency as shown in Table~\ref{tab:scaling-amm}.

Throughput-wise, we record a low throughput of 0.42 tx/s to 33.04 tx/s for a daily volume of 50K to 5M transactions (roughly 1x-100x Uniswap's daily volume). This is mainly due to the mainchain blocks not being full as this workload is way below the capacity that \sysname can handle. While for traffic that is 500x Uniswap's daily volume, \sysname achieves a throughput of 138.06 tx/s.

Latency-wise, we achieve a quasi-instant trend when the daily volume is 50K-5M (transactions that arrive at the beginning of the round get processed within the same round, while the residual amount of latency is due to transactions generated close to the epoch end and processed in the next epoch). This leads to payouts being processed within one epoch, with the transactions happening earlier in the epoch having a higher payout latency than the ones happening towards the end, and averaging at half an epoch plus the time needed to confirm the $\sync$ transaction on the mainchain $(\sim 120s)$. Transaction congestion happens when the daily volume is 25M, resulting in higher average latency as the table shows.

\begin{table}[t!]
    \caption{Scalability of \sysname.}
    \label{tab:scaling-amm}
    \centering{}
     \fontsize{9pt}{11pt}\selectfont
\resizebox{\columnwidth}{!}{
    \begin{tabular}{|l|c|c|c|c|}
    \hline
        \textbf{Daily volume} & 50K & 500K & 5M  & 25M \\ 
    \hline
        \textbf{Throughput (tx/s)} &  0.42 & 3.41 & 33.04 & 138.06\\
    \hline
        \textbf{Avg. sc latency (s)} &7.13& 7.13 & 7.13 & 231.52 \\
    \hline 
        \textbf{Avg. payout latency (s)} & 120.71 & 120.71  & 120.71   & 346.49 \\
    \hline
    \end{tabular}
    }
\end{table}

\subsection{Comparison with Rollups} We compare \sysname to an AMM that uses an Optimism-inspired rollup solution (\ammop). \ammop processes 1.8MB of transactions per batch~\cite{optimism-size}, and takes around 35 sec to process a batch (Optimism takes 2-4 Ethereum rounds~\cite{jumpcryptoBridgingFinality}, we average that to 3 rounds). \sysname's configuration and traffic generation patterns remain the same as above. Table~\ref{tab:amm-op-comparison} shows our results, where transaction latency is the time needed for a transaction to appear in a meta-block in \sysname, or a non-finalized processed rollup
in \ammop.

\begin{table}[t]
    \centering

    \caption{Comparison between \sysname and \ammop.}
    \label{tab:amm-op-comparison}

    \begin{tabular}{|l|C{1.85cm}|C{1.85cm}|C{1.85cm}|}
        \hline
         &  Throughput  & Transaction  & Payout\\
        &      (tx/s)  &   Latency (s) & Latency (s) \\
         \hline
 \ammop  &           51.16         &       2577.28                 &     604,815.28              \\
        \hline
 \sysname &          138.06         &           231.52             &                346.49    \\
        \hline
    \end{tabular}
\end{table}

As shown, \sysname outperforms \ammop across all metrics. We observe a 2.69x improvement in throughput and a 91.02\% decrease in transaction latency. This is due to the higher capacity of \sysname compared to \ammop within the same time duration; \sysname processes 5 MB of transactions vs 1.8 MB for \ammop, allowing for more transactions to be processed in a small time duration, and hence, reducing waiting time for transactions. Finally, we observe that the payout latency is abysmal for \ammop, due to the long contestation period (which is one week in Optimism).

\section{Conclusion}
\label{sec:conclusion}
We presented \sysname, a secure state growth controller and throughput booster for AMMs. \sysname combines a dependent-sidechain architecture with a functionality split of the AMM between the main and the side chains, and introduces several techniques to allow layer-2 based processing of the AMM workload without correctness and security. We analyze the security of our system and conduct thorough performance evaluation experiments. The results show the great potential of \sysname in reducing the on-chain storage footprint of AMMs and boosting their scalability.

\section*{Acknowledgment}
M.E.N. is supported by NSF under Grant No. CNS-2226932, and G.A. is supported by the Latest in DeFi Research (TLDR) fellowship funded by Uniswap Foundation.

\bibliographystyle{IEEEtran}
\bibliography{ammBib}

\appendices
\section{Overview of PBFT-based Consensus and Committee Election}
\label{apdx:pbft-review}
Given the low throughput and high confirmation delays of blockchains, several consensus protocols utilize practical Byzantine fault tolerant (PBFT)~\cite{castro1999practical} to boost throughput and speed agreement on new blocks. These protocols, in general, employ the following paradigm: the system operates in rounds and epochs, where a round is period during which a new block is mined, and an epoch is $\omega$ consecutive rounds. For each epoch, a committee of the blockchain miners is elected to mine blocks during each round. This mining is basically reaching an agreement on a new block proposal. Generally, a committee member agrees by voting on (i.e., signing) the block. Owning a mining power, e.g., a stake or computing power, is used to establish Sybil-resistant identities for the miners to defend against attackers who may spawn Sybils to take over the committee. Furthermore, the committee is refreshed for every epoch to reduce the chances that an adversary takes over the majority of the committee over time (i.e., to mitigate targeted attacks against the committee).

There are two flavors of PBFT-based consensus: leader-based consensus, e.g.,~\cite{Kogias16,kwon2014tendermint, buchman2016tendermint, yin2019hotstuff, moniz2020istanbul, gueta2019sbft, kotla2007zyzzyva}, and voting-based consensus, e.g.,~\cite{Gilad17, miller2016honey, antoniadis2023leaderless, crain2018dbft}. In leader-based consensus, a leader is elected for each committee, so that in each round this leader proposes a block and initiates the agreement to collect votes from the committee members on the block proposal. In voting-based consensus, committee members are divided into block proposers and voters. In any round, each proposer proposes a block that voters vote on, and then the block proposal with a vote majority will be added to the blockchain. As noted, different from proof-of-work blockchains, where mining a block requires solving a computationally-heavy puzzle, in PBFT-based consensus mining a new block relies on reaching agreement. Each committee member validates a block proposal (to ensure that it contains valid transactions and state changes) and, if valid, signs the block to indicate their agreement. 

In both flavors, secure committee election is crucial. This election algorithm must satisfy several properties including: \emph{Sybil-resistance} to ensure that only real miners are elected. This is usually done, as mentioned above, by requiring miners to own some mining power so that the probability of a miner to be elected is proportional to the amount of mining power they own. For example, in proof-of-work blockchains, this probability is proportional to the computing power this miner owns, as in ByzCoin~\cite{Kogias16} that employs a sliding window approach for election. In proof-of-stake blockchains, this will depend on the amount of stake a miner owns, as in Algorand~\cite{Gilad17} and Omniledger~\cite{Kokoris18}. Thus, the Sybil-resistance property also captures fairness in election so a party is elected based on her share in the system. 

\emph{Unpredictability} is another essential property to address targeted attacks resulting from knowing the committee members in advance. Cryptographic sortition~\cite{Gilad17,Kokoris18} supports that, where election is based on a Verifiable Random Function (VRF). \emph{Public verifiability} is another important property to enable verifying that indeed miners have been elected and accept their proposals/voting (or agreement endorsement). For example, in cryptographic sortition, the VRF produces a publicly verifiable proof attesting to the output (where this output determines whether a miner is elected). Various committee election mechanisms with different trade-offs exist in the literature~\cite{Kogias16,pass2017hybrid,Gilad17,kiayias2017ouroboros,gavzi2023fait}.

\section{Security Analysis}
\label{apdx:sec-analysis}
We aim to develop a secure state growth control solution that preserves the valid and secure operation of the AMM. \sysname builds a sidechain, which is basically a blockchain, that interacts with the application layer of the mainchain through the base AMM smart contract $\tokenbank$. Thus, starting with an AMM deployed on a secure mainchain (that satisfies safety and liveness as defined below), we want to show that deploying \sysname does not violate the correct and secure operation that the AMM had when it was completely deployed on the mainchain.

\emph{Ledger security.} A ledger $\led$ is secure if it satisfies the following properties~\cite{garay2015bitcoin}:
\begin{description}
\item[Safety:] For any two time rounds $t_1$ and $t_2$ such that $t_1 \leq t_2$, and any two honest parties $P_1$ and $P_2$, the confirmed state of $\led$ (which includes all blocks buried under at least $k$ blocks, where $k$ is the depth parameter) maintained by $P_1$ at $t_1$ is a prefix of the confirmed state of $\led$ maintained by party $P_2$ at time $t_2$ with overwhelming probability.

\item[Liveness:] If a valid transaction $tx$ is broadcast at time round $t$, then with overwhelming probability it will be recorded on $\led$ at time at most $t + u$, where $u$ is the liveness parameter.
\end{description}

As such, a secure blockchain grows over time and records only valid transactions and blocks in an immutable way, i.e., confirmed blocks and transactions cannot be altered, with a consistent view of the confirmed chain among the miners. The ledger protocol is parameterized by predicates to verify transaction and block validity. For dApps, validity is governed by the code of their smart contracts, and miners ensure that the ledger state changes have been produced by a successful execution of this code. 

\sysname reduces the AMM functionality deployed on the mainchain, and it processes most of the workload (following the same logic of the original AMM) on the sidechain. Proving security of \sysname boils down to preserving the following invariant: the underlying AMM still achieves safety and liveness. In other words, the security and correctness guarantees of an AMM deployed on the mainchain must be preserved after deploying \sysname. Note that safety of an AMM means reaching a consistent view among the miners about its confirmed correct state, and liveness means that this state grows over time by adding correct records.

This amounts to showing that \sysname's sidechain is a secure ledger, and that all additional techniques that \sysname introduces preserve the invariant above. Recall that the sidechain implements a secure PBFT consensus protocol under the assumption that any committee in any epoch can have at maximum $f$ faulty members (out of $3f+2$ members). So it also achieves ledger safety and liveness. Thus, what remains to show is that \sysname operational design---meta-block pruning, mainchain (i.e., $\tokenbank$) state syncing, and handling interruptions---does not violate the invariant above.

Briefly, having a secure PBFT consensus, that adopts the leader change to deal with malicious/unresponsiveness committee leaders, on the sidechain guarantees safety and liveness of the sidechain. Thus, only correct records that follow the underlying AMM logic are accepted. The sidechain also satisfies public verfiability since meta-blocks do not get pruned until their sync-transaction is confirmed on the mainchain. Safety and liveness of the AMM base smart contract are guaranteed by the security of the mainchain. For the syncing process, which is the connection between AMM state changes produced on the sidechain and its mainchain state, its security relies on sidechain/mainchain consensus security, as well as the security of the TSQC authentication method. The former guarantees that syncs are produced correctly at the end of each epoch, and will be processed based on $\tokenbank$ code on the mainchain. While the latter guarantees that syncs are accepted only if produced by the rightful epoch committee. Thus, breaking the security of syncing leads to breaking the security of consensus, or the security of TSQC authentication (which relies on the security of committee election and threshold signatures). 

In order to prove Theorem~\ref{th:security}, we prove two lemmas showing that \sysname preserves safety and liveness of the underlying AMM (these proofs are inspired by those found in~\cite{chainboost-paper}).

\begin{lemma}\label{lm:safety}
\sysname preserves the safety of the underlying AMM.
\end{lemma}

\begin{proof}
Since \sysname offloads AMM operations to the sidechain, and implements meta-block pruning as well as mainchain (i.e., $\tokenbank$) state syncing, we identify the following threats that may impact safety in our system:

\begin{itemize}
    \item Invalid processing of AMM transactions (or violating sidechain quality): the sidechain committee does not follow the AMM logic in processing transactions, or accept transactions from users who do not own enough deposits, or process these transactions based off an invalid initial state of the token pool. This leads to mining invalid meta and summary blocks. 
    
    \item Out-of-sync AMM state on the mainchain: A committee leader does not issue a $\sync$ function call at the end of the epoch, or the committee does not conclude agreement on an issued $\sync$, causing the state of the AMM on the mainchain and the state maintained on the sidechain to be out of sync. Out-of-sync AMM mainchain state could also happen due to rollbacks on the mainchain causing recent $\sync$ calls to be lost.
    
    \item Invalid syncing: A sidechain committee agrees on invalid inputs (or syncing information) for the $\sync$ function call, or an illegitimate committee pretends to be the elected one and issues such an invalid sync.
\end{itemize}

We show how \sysname mitigates these threats, which mainly relies on the use of a secure PBFT consensus with a secure committee election mechanism, a secure threshold digital signature scheme, the leader-change mechanism to handle the case of malicious/unresponsive committee leader, and pruning meta-blocks only after the corresponding sync-transaction is conformed on the mainchain.

\emph{Invalid processing of AMM transactions.} The sidechain in \sysname uses the same logic adopted by the underlying AMM to process all transactions. This is based off the latest state of the AMM (with respect to user deposits, liquidity positions, and pool balances from $\tokenbank$ on the mainchain). The use of a secure PPFT protocol guarantees that only valid records that conform with these rules will be accepted, and agreed-upon, in meta-blocks. Also, processing transactions and verifying a meta-block will be based on the latest pool state that each member of the committee retrieves from the mainchain. Security of the PBFT protocol also guarantees that the committee only agrees on valid summary-blocks based on the summary rules in \sysname and these meta-blocks. As meta-blocks do not get pruned until their $\sync$ call transaction is confirmed on the mainchain, anyone can verify the validity of these blocks and the validity of the summary-block (as well as the $\sync$ call).  

Moreover, even if a leader proposes invalid blocks, due to using a secure PBFT-based consensus, the committee will not agree on such blocks, and without this agreement, such blocks will not be adopted on the sidechain. In this case, the committee will not only reject the block proposal, but also initiate a leader-change to elect a new leader who will take over for the rest of the epoch.


\emph{Out-of-sync AMM state on the mainchain.} This is mitigated using regular syncing as well as mass-syncing. On the one hand, security of the sidechain, via its PBFT-based consensus, guarantees that summary-blocks are produced correctly with corresponding valid sync-transactions at the end of each epoch. Also, since only up to $f$ nodes in the committee can be malicious, out of $3f+2$, and $2f+2$ votes are needed to reach agreement, an adversary cannot control the majority of the committee to stall agreement on a sync call. Security of the mainchain guarantees that valid syncing will be accepted and processed.

On the second hand, a leader who does not initiate an agreement on the $\sync$ function call, or does not submit the result of the agreement to $\tokenbank$, is easily detected by the new committee as no function call has been issued and recorded on the mainchain. The committee of the next epoch then syncs $\tokenbank$ based on the summaries it produces in its epoch and the ones in the previous (one or multiple) epoch, i.e., perform mass-syncing. Same for any rollbacks that may happen on the mainchain, mass-syncing will include all summaries that has been lost due to the rollback. Since a new committee is elected for every epoch, it is infeasible for an adversary to control every committee leader on the long run.

\textit{Invalid syncing.} If a leader issues invalid $\sync$ call inputs, the sidechain committee will not endorse these inputs (due to employing a secure PBFT-based consensus). So simply this invalid syncing inputs will be ignored, and mass-syncing (discussed above) will handle the syncing within the next epoch. Having an illegitimate committee pretend to be the rightful one to submit (invalid) sync is addressed using the TSQC-based authentication in \sysname. The committee of epoch $e$ will not accept the generated committee public verification key $\vk_c$ unless there are valid proofs of election showing that the newly claimed committee is the rightful one (so this relies on the secure committee election whose output is publicly verifiable). Furthermore, an illegitimate committee (or an attacker) instead may try to forge a signature over the syncing information under a valid $\vk_c$, which succeeds with negligible probability by the security of the threshold digital signature scheme.

Accordingly, \sysname satisfies the safety of the AMM in the sense that all sidechain workload is processed correctly, and the mainchain state of the AMM is synced correctly.
\end{proof}

\begin{lemma}
\sysname preserves the liveness of the underlying AMM.
\end{lemma}

\begin{proof}
The liveness of the sidechain impacts the liveness of the AMM in the sense that any liveness threats on the sidechain will impact the operation progress of the AMM, i.e., processing and publishing newly issued transactions. We identify the following threats that may arise and violate the liveness of the AMM under our setting:
\begin{itemize}
    \item Denial of service (DoS) attacks: The sidechain committee deliberately ignores and omits transactions coming from particular users.
    \item Violating sidechain liveness: the sidechain committee does not mine meta- and summary-blocks or does not submit $\sync$ function calls (that could be due to malicious/unresponsive leader or malicious/unresponsive committee that does not reach an agreement).
    \item Violating the public verifiability of the sidechain: this covers all threats related to the syncing and pruning of meta-blocks that may impact the public verifiability of the sidechain (which in turn impacts the public verifiability of the AMM state). Under these we have AMM users who may claim that they they are supposed to obtain larger payout than what is dispensed, or more liquidity shares/fees that what is reported in their liquidity positions (by claiming that they submitted valid transactions for that that were pruned). 
\end{itemize}

Proving that \sysname mitigates these threats is the same as in~\cite{chainboost-paper}. For completeness, we review the proof arguments here. \emph{DoS} is addressed via rotating committee election (a new committee is elected for each new epoch) such that this committee has an honest majority. A leader that targets particular users, and so omit their transactions from all proposed meta-blocks, will operate for one epoch and then a new committee with a new leader will take over for the next epoch. Thus, maintaining a situation where all future leaders are malicious and perform the same DoS is infeasible. \emph{Sidechain liveness} is satisfied due to the use of a secure PBFT-based consensus (which achieves liveness). Also, leader-change allows changing a leader who deliberately attempt to stall the network by not proposing new meta-blocks. While mass-syncing mitigates the case of a malicious leader who does not initiate agreement on a summary-block or $\sync$ function call. 

\emph{Public verifiability} is guaranteed by the security of the PBFT consensus; only valid meat-blocks and summary-blocks are agreed upon that capture the processed transactions and state changes correctly. Also, meta-blocks are not pruned until their corresponding $\sync$ call is confirmed on the mainchain, allowing anyone to verify the syncs, and summary-blocks are permanent (as discussed in the proof of Lemma~\ref{lm:safety}). Also, by having the AMM base smart contract $\tokenbank$ on the mainchain synced correctly based on the sidechain summaries (also as discussed in the proof of Lemma~\ref{lm:safety}), summaries are not lost. All of these allow anyone to verify the validity of the evolving state of the AMM. 

Accordingly, \sysname preserves the liveness of the AMM.
\end{proof}

\section{Concrete Use Case: Uniswap} 
\label{sec:uniswap}
Uniswap has three versions: Uniswap V1, released in November 2018, consisted of the baseline protocol that implemented ERC20 token swaps with Ethereum and all of the liquidity management methods (mint, burn, collect). Uniswap V2, released in August 2020, introduced ERC20 to ERC20 swaps, liquidity provision incentives, and oracles. And Uniswap V3 released in May 2021, introduced concentrated liquidity, a nonfungible representation of liquidity positions, and further improvements to the oracle systems. Uniswap adopts the constant product formula for computing the trading price described in Section~\ref{sec:background}. Uniswap is among the most popular AMMs in practice and commands a large market share in the AMM industry. In this section, we provide an overview of the set of contracts that implement the Uniswap functionality on Ethereum, and the execution trace of the supported transactions.

\subsection{Uniswap Supporting Contracts}
Based on the Uniswap documentation~\cite{uniswap-docs} and its reference implementation~\cite{uniswap-ref-impl}, Uniswap is implemented as a set of five contracts: $\PoolDep$, $\PoolFac$, $\NFPM$, $\NFTPD$, and $\SwapRouter$.\medskip

\noindent\textbf{Pool factory and deployer.}
The $\PoolFac$ and $\PoolDep$ contracts are responsible for setting up new token pools. $\PoolDep$ provides the interface, and $\PoolFac$ creates the actual pool. Once a pool is created, clients and LPs can start interacting with it.\medskip

\noindent\textbf{Nonfungible position manager and token descriptor.}
These contracts manage the liquidity positions by handling processes associated with minting, collecting, and burning/adjustment of liquidity positions. The $\NFPM$ contract serves as a "pit stop" for an LP's input tokens, such that the LP deposits input tokens for mint transactions before executing the mint functionality of a particular pool. This intermediate step allows Uniswap to guarantee that the input tokens will actually be delivered by the LP, as they are deposited in the first step, and automatically retrieved from the $\NFPM$ contract by the pool contract when needed. The LP can later retrieve any tokens not used by the mint transaction. These two contracts also implement a unique NFT-based identifier for liquidity positions such that LPs can trade positions amongst themselves.\medskip

\noindent\textbf{Swap router.}
The $\SwapRouter$ contract manages the swapping process. It implements functions such as $\ExIn$ and $\ExOut$ to facilitate specific kinds of swaps. The $\SwapRouter$ also serves as a "pit stop" for input tokens, requiring clients to deposit tokens they want to trade before performing swap transactions.\medskip

There are additional smart contracts deployed in the Uniswap ecosystem, e.g., lens contracts which act as an on-chain oracle to record the price and liquidity history of a given pool. We do not provide further information about such contracts since we focus on the core functionality of Uniswap in our usecase implementation.

\subsection{Transaction Execution Trace}
The core transaction types supported in Uniswap---swaps, mints, burns, collects, and flashes---are executed as follows.\medskip

\noindent\textbf{Swap.} Regardless of the type of swap (exact in/out), clients must first deposit their input tokens in the $\SwapRouter$ contract and approve it to spend their tokens. The client then calls the relevant function of the $\SwapRouter$ contract ($\ExIn$ or $\ExOut$) to submit a swap transaction. If the user is performing an $\ExOut$ swap, they should implement an additional set of conditional transfers to occur after the call to $\ExOut$ to retrieve unspent input tokens. 
    
Internally to either function, the pool's swap function is called ($\SwapTx_{pool}$). $\SwapTx_{pool}$ determines the price of the swap, distributes the liquidity provider fee across the positions whose liquidity is used to fill in the swap, and transfers the output tokens to the client before invoking the $\SwapCallback$ function. $\SwapCallback$ is called to retrieve the needed amount of input tokens from the $\NFPM$ contract. The client's contract can now call any additional transfers to retrieve unspent input tokens (in the case of $\ExOut$).\medskip

\noindent\textbf{Mint.} The user first creates a smart contract capable of receiving ERC721 tokens. This contract must implement the following: a method to receive and store the nonfungible position tokens, and another method to execute the mint. Alternatively. the user may simply forgo the ERC721 receiver contract, allowing their NFT positions to remain as part of the NFTPM contract. Should a user decide to do this, they can simply invoke the same relevant functions of the NFTPM below by using any library which allows interfacing with smart contracts. Mint execution encompasses the following:
\begin{itemize} 
    \item LP transfers their input tokens to the \\$\NFPM$ contract, and authorizes it to spend their tokens when executing the mint. 
    \item LP calls the $\NFPM$ contract's mint method ($\MintTx_{NFPM}$), which creates the NFT position structure, and then calls $\addLiquidity$ to create the liquidity position. 
    \item The $\addLiquidity$ function retrieves the relevant pool from the passed tokens and fee tier. Using the current price ratio, and the desired amount of tokens to be added as liquidity (passed by the user), an applicable liquidity share is calculated using the function $\mathsf{getLiquidityForAmounts}$. Then, the pool's mint function $\MintTx_{pool}$ is called, passing the liquidity value computed by $\mathsf{getLiquidityForAmounts}$. 
    \item $\MintTx_{pool}$ creates a new liquidity position with the liquidity share provided, and outputs the exact amount of the token pair required for the position. This liquidity position contains the amount of liquidity provided by the LP, the price range within which this liquidity can be used for token swaps, and state variables keeping track of the fee changes incurred by swapping the token pairs and the fees owed to the liquidity provider. Then it calls the $\MintCallback$ function, after which the amount of each token to be added to the pool is returned to the $\NFPM$ contract. 
    \item The $\MintCallback$ function verifies that the caller is a valid pool contract, and then transfers the used input tokens from the $\NFPM$ contract to the pool contract.
    \item $\MintTx_{NFPM}$ returns the position NFT, as well as the amount of each token actually added to the pool. 
    \item LP can retrieve any unspent input tokens from the $\NFPM$ contract; this is why the LP needs to implement a method to receive and store an ERC721 token in their contract.
\end{itemize}    

\vspace{3pt}
\noindent\textbf{Collect.} LP calls $\NFPM$'s collect method ($\CollectTx_{NFPM}$) to execute a collect transaction, passing the amount of fees they wish to withdraw along with the nonfungible position token's ID representing their liquidity position. The function $\CollectTx_{NFPM}$ verifies that the transaction issuer is indeed the position owner, and then identifies the target pool based on the token ID. After that, it retrieves the current token amount owed to the owner from the position through the fee calculation process. The latter is an optimization introduced in Uniswap V3 to accommodate for concentrated liquidity positions and to reduce the overall gas usage. Specific details on the calculation process for fees in Uniswap V3 can be found in its whitepaper~\cite{Uniswap-V3-whitepaper}.\medskip

\noindent\textbf{Burn.} To burn a position, the following steps take place:
\begin{itemize}
    \item LP withdraws all tokens owned by that position by first calling $\decreaseLiquidity$ from the $\NFPM$ contract. This function retrieves the relevant pool contract, and calls the pool's Burn function.
    \item The burn function takes the requested amount of tokens to burn, calculates the actual share of liquidity owned by the position which can be burnt (up to the requested amount), and decrements the calculated amount from the positions owned tokens. Finally, it adds the decremented amount to the liquidity positions owed-tokens metric, such that they can be withdrawn by invoking collect.
    \item Once LP have decreased the liquidity owned by their position to zero, they can invoke collect to retrieve those funds before calling $\BurnTx_{NFPM}$.
    \item $\BurnTx_{NFPM}$ checks that the passed liquidity position does not own any shares of liquidity and all owed tokens have been collected. Should these checks pass it deletes the liquidity position and the NFT associated with it.
\end{itemize}

\vspace{3pt}
\noindent\textbf{Flash.} To execute a flash transaction, the client begins by writing and deploying a smart contract which overwrites the $\FlashTx_{callback}$ method of the liquidity pool. The $\FlashTx_{callback}$ function is responsible for paying back the loan. As such, should a client want to perform arbitrage with the loan, they begin by overwriting $\FlashTx_{callback}$, simply adding in solidity code to execute their arbitrage opportunity. They can then call the flash function of the liquidity pool from which they would like to execute a flash transaction. $\FlashTx$ itself simply transfers the requested loan of tokens to the client, where they are used for the arbitrage opportunity, before being re-transferred to the pool, plus the associated fees, by $\FlashTx_{callback}$. Should the arbitrage prove non-profitable, or the contract fail to pay back the flash loan for any reason, the entire transaction is concluded, resulting in the pool never having transferred the loan in the first place. This is possible due to the entire flash process occurring in a single Ethereum transaction. 

\begin{remark}[On NFT-based liquidity positions]
Uniswap V3 introduced an NFT-based approach, using an ERC721 wrapper, to track ownership of liquidity positions. This approach allows for a streamlined process for the verification and transfer of ownership of a position. This can be also adopted in \sysname. At a high level, $\tokenbank$ can be extended to support the NFT approach by utilizing the same implementation found in Uniswap. The caveat though is that creating an NFT will wait until the end of the epoch since it requires mainchain operation (now in \sysname positions can be created immediately and synced back to the mainchain since tracking owenership relies on the LPs' public keys and identifiers are generated at random). Thus, any operations on these new positions has to wait until the next epoch after creating the position NFT.
\end{remark}

\section{Uniswap Traffic Analysis}
\label{appdx:uniswap-traffic-analysis}
In order to find the volume of each transaction type, we used the following query on Dune analytics uniswap\_v3\_ethereum dataset (the following citation is a direct link to the query used~\cite{UniswapTrafficByTxType}). The query retrieves and counts all of the transactions happened since 2019, splitting them by year and transaction type. uniswap\_v3\_ethereum is one of Dune analytics "decoded projects", meaning that it is a dataset formed from the ABI of the smart contracts that operate the protocol in question. Once a user submits a contract for decoding, Dune uses the ABI to generate a table of transactions that is query-able by function call or event. As such, the uniswap\_v3\_ethereum is a set of tables which contain the decoded smart contracts that constitute the Uniswap V3 protocol. Since we are interested in transaction volumes, the above query counts all transactions since 2019 by pulling any row in the tables uniswap\_v3\_ethereum.Pair\_call\_burn, uniswap\_v3\_ethereum.Pair\_call\_collect, etc., for which the \\call\_block\_time is $\geq$ 01/01/2019.\medskip 

\noindent\textbf{Traffic distribution or transaction type frequency.} 
We calculate the frequency of each transaction type by computing the number of transactions of that particular type divided by the total number of transactions from all types. The volume of transactions is gathered by the query above. The frequency shown in Table~\ref{tab:unsiwap-traffic} is calculated for the year 2023.\medskip 

\noindent\textbf{Number of trades per 24 hours by transaction type (volume).}
We calculate this metric by taking the 2023 yearly total transaction count (found by the query above), and then compute the average daily volume of each transaction based on the frequency computed above. The results can be found in Table~\ref{tab:unsiwap-traffic}.\medskip 

\noindent\textbf{Transaction sizes.}    
In order to collect the average size of each transaction type, we implemented a python script to interact with an Ethereum node hosted by chainstack~\cite{Chainstack}. We first collected the transaction hashes of a sufficient amount (approximately 40,000 swaps and 10,000 of all other transaction types) of each transaction type from the Uniswap V3 subgraph~\cite{UniswapV3Subgraph} (this can similarly be done with the dune query provided above, by modifying the select on line 37 to include the transaction hash). Then, we ran a script to analyze the collected transactions. This script basically iterates through the json file which contains our aggregated transaction hashes. For each transaction hash, it performs a web3.eth.get\_raw\_transaction query to retrieve the full raw transaction size. After that, it computes the average size for each transaction type (which is basically sum of total size divided by the number of transactions). The results are also found in Table~\ref{tab:unsiwap-traffic}.

\begin{table}[t!]
\caption{Transaction type breakdown in Uniswap traffic for the year 2023.}
\label{tab:unsiwap-traffic}
\begin{tabular}{ |p{0.16\columnwidth}||p{0.23\columnwidth}|p{0.2\columnwidth}|p{0.2\columnwidth}|}
 \hline
 Transaction Type& Percent of all traffic & Volume per 24 hr & Average Size (B) \\
 \hline
 Swap & 93.19 \% & 52,379 & 1007.83\\
 Mint & 2.14 \% & 1,204 & 814.49\\
 Burn &  2.38 \% & 1,338  & 907.07\\
 Collect & 2.27 \% & 1,275 & 921.80\\
 \hline
\end{tabular}
\end{table}

\section{Additional Evaluations: Impact of Parameter Configuration}
\label{appdx:more-eval}
In this section, we study the impact of parameter configuration on \sysname's performance, including traffic distribution, block size, sidechain round duration, and number of rounds per epoch.

\vspace{3pt}
\noindent\textbf{Impact of block size.} We test the impact of the sidechain block size with the goal of finding an optimal block size for our system. Thus, we compare different deployments of \sysname with different block sizes against Uniswap's Sepolia deployment. We run the protocol with the following block sizes \{$0.5, 1, 1.5, 2$\} MB, and we increase the daily volume to 50M transactions. We measure the impact on throughput and transaction/payout latency with the goal of identifying the block size that maximizes throughput while minimizing latency. Our results can be found in Table~\ref{tab:blocksize-amm}.

\begin{table}[t!]
    \caption{Impact of different sidechain block sizes.}
    \label{tab:blocksize-amm}
    \centering
      \fontsize{9pt}{11pt}\selectfont
    \begin{tabular}{|p{0.35\columnwidth}||p{0.1\columnwidth}|p{0.1\columnwidth}|p{0.1\columnwidth}|p{0.09\columnwidth}|}
    \hline
        \textbf{BlockSize (MB)} & \textbf{0.5} & \textbf{1} & \textbf{1.5} & \textbf{2} \\ 
    \hline
        \textbf{Throughput (tx/s)} &  68.97  & 138.61 & 207.52 & 276.43\\
    \hline
        \textbf{Avg. sc. latency (s)} & 4357.00 & 1603.01 & 687.98 & 230.48\\
    \hline 
        \textbf{Avg. payout latency (s)} &4472.63&  1719.10&  804.05  &  345.44\\
    \hline
    \end{tabular}
\end{table}

As expected, increasing the block size improves both throughput and latency, as more transactions can be packed in a block which reduces queue congestion. However, larger block sizes mean a larger propagation delay that could be problematic for short sidechain round duration. Thus, system designers should be careful when choosing an optimal sidechain block size, balancing between the block size that can handle the daily volume of transactions while capturing the intricacies of large network transfers. 

\vspace{3pt}
\noindent\textbf{Impact of sidechain round duration.} Another important factor to study is the impact of the sidechain round duration. An ideal round duration should allow for consensus to conclude while maximizing throughput and minimizing latency. As PBFT agreement takes on average around 6 sec to conclude in our implementation, we test the following round duration values: 7, 11, 16, and 21 seconds, and report the performance metrics as before (Table~\ref{tab:roundduration-amm}).

\begin{table}[t!]
    \caption{Impact of different sidechain round durations.}
    \label{tab:roundduration-amm}
          \fontsize{7pt}{9pt}\selectfont
        \resizebox{\linewidth}{!}{
    \centering

    \begin{tabular}{|l|c|c|c|c|}
    \hline
        \textbf{Sc round duration (s)} & 7 & 11 & 16 & 21 \\ 
    \hline
        \textbf{Throughput (tx/s)} &  138.06 & 92.18 & 61.75 & 46.31\\
    \hline
        \textbf{Avg. sc latency (s)} & 231.52 & 921.64 & 1950.92 & 2975.90\\
    \hline 
        \textbf{Payout latency (s)} &346.49 & 1087.95& 2193.85 &  3295.11\\
    \hline
    \end{tabular}
}
\end{table}

Throughput-wise, we observe that as the block time increases, throughput decreases and latency increases. This is due to processing the same amount of transactions while increasing the time needed to produce a block. To choose an optimal block time, system designers need to take into consideration the time required for the sidechain consensus and network propagation delays, while aiming to generate new blocks as fast as possible.

\vspace{3pt}
\noindent\textbf{Impact of the number of sidechain rounds per epoch.} We test the impact of the number of sidechain rounds within an epoch. The goal is to find an epoch length that maximizes throughput, and minimizes transaction/payout latency, based on the optimal sidechain round duration from the experiment above. Thus, we pick our epoch to have $\{5, 10, 20, 30, 60, 96 \}$ sidechain rounds, each of which lasts 7 sec, and report the performance metrics (Table~\ref{tab:epochsize-amm}).

Having short epochs negatively affects throughput and the sidechain latency. As a matter of fact, frequent summary-blocks harm performance since this leads to a larger number of $\sync$ calls that are costly. At the same time, fewer transactions are processed within the epoch, thus affecting both latency and throughput. Longer epoch duration reduces the sidechain latency and increases throughput. However, this affects the payout latency adversely since $\sync$ calls now are much fewer and spaced out (so users have to wait longer, as the epoch itself is longer, to obtain their actual token payouts). Also, this means that these users have to put larger deposits to cover their (long) epoch activities, which could be undesirable. Based on our results, we achieve the best payout latency when the epoch lasts for 20 sidechain rounds, which is equivalent to 140 sec.

\vspace{3pt}
\noindent\textbf{Impact of traffic distribution.} In this experiment, we evaluate different traffic distributions as follows (all the numbers are percentages):  $(s,m,b,c) \in \{(60, 20, 10, 10)$, $(60, 10, 20, 10),$ $(60, 10, 10, 20),$ $(80, 10, 5, 5),$ $(80, 5, 10, 5),$ $(80, 5, 5, 10) \}$, where $(s, m, b, c)$ stand for swaps, mints, burns, and collects, respectively. As noted, in these configurations we keep the swap operations dominant to align with the baseline AMM traffic distribution observed in practice (see Appendix~\ref{appdx:uniswap-traffic-analysis}). Our results can be found in Table~\ref{tab:distribution-amm}.

\begin{table}[t!]
    \caption{Impact of number of sidechain rounds per epoch.}
    \label{tab:epochsize-amm}
    \resizebox{\linewidth}{!}{
    \centering
     \fontsize{9pt}{11pt}\selectfont
    \begin{tabular}{|l|c|c|c|c|c|c|}
    \hline
        \textbf{Epoch len} & \multirow{2}{*}{5} & \multirow{2}{*}{10} & \multirow{2}{*}{20} & \multirow{2}{*}{30} & \multirow{2}{*}{60} & \multirow{2}{*}{96} \\ 
        \textbf{(sc rounds)}                      &  & & & & &  \\
    \hline
        \textbf{Throughput} &  \multirow{2}{*}{114.27} & \multirow{2}{*}{128.53} & \multirow{2}{*}{135.90} & \multirow{2}{*}{138.06} & \multirow{2}{*}{140.66} & \multirow{2}{*}{141.53}\\
           \textbf{(tx/s)}                      &  & & & & &  \\      
    \hline
        \textbf{SC latency (s)}  & 517.94 & 333.54 & 255.57 & 231.52 & 208.96 & 199.55 \\

    \hline 
        \textbf{Payout} & \multirow{2}{*}{545.12} & \multirow{2}{*}{337.86} & \multirow{2}{*}{334.81} & \multirow{2}{*}{346.49} & \multirow{2}{*}{434.94}  & \multirow{2}{*}{546.04}\\
        \textbf{latency (s)}                      &  & & & & &  \\
    \hline
    \end{tabular}
    }
\end{table}

\begin{table}[t!]
    \caption{Impact of traffic distribution.}
    \label{tab:distribution-amm}
    \centering
    \resizebox{\linewidth}{!}{
    \begin{tabular}{|l|c|c|c|c|c|c|}
    \hline 
        \textbf{Swap \%} & \multicolumn{3}{|c|}{60} & \multicolumn{3}{|c|}{80} \\
    \hline
        \textbf{Mint \%} & 20 & 10 & 10 & 10 &  5& 5 \\ 
        \textbf{Burn \%}              &      10  & 20  & 10  &  5 &  10 & 5  \\
        \textbf{Collect \%}            &    10      & 10  & 20 & 5 &  5 &  10 \\
    \hline
        \textbf{Throughput} & \multirow{2}{*}{145.16}  & \multirow{2}{*}{143.76} & \multirow{2}{*}{140.91} & \multirow{2}{*}{143.76} & \multirow{2}{*}{140.23} & \multirow{2}{*}{140.14} \\
         \textbf{(tx/s)}                      &  & & & & &  \\    
    \hline
    \textbf{SC latency (s)} & 162.26 & 175.35&  177.39 & 202.48 & 215.06 & 210.35 \\  
    \hline 
        \textbf{Payout} & \multirow{2}{*}{277.99}& \multirow{2}{*}{291.05} & \multirow{2}{*}{293.03} & \multirow{2}{*}{317.23} & \multirow{2}{*}{329.81} & \multirow{2}{*}{324.43} \\
        \textbf{latency (s)}                      &  & & & & &  \\     
    \hline
     \textbf{Max sc} &  \multirow{2}{*}{31831} & \multirow{2}{*}{31831} & \multirow{2}{*}{31831}  & \multirow{2}{*}{31831}  & \multirow{2}{*}{31831}   & \multirow{2}{*}{31831}  \\
    \textbf{growth (B)}                      &  & & & & &  \\     
    \hline
    \end{tabular}
}
\end{table}

When varying the traffic distribution, the metrics we report remain similar. This is because transaction sizes are very close, this yields blocks containing approximately the same number of transactions, regardless of the transaction distribution. As for the maximum chain growth, it is bounded by the number of users participating during an epoch and the number of positions they create. Thus, it remains invariant even with a variation of transaction distributions since the number of users is the same.

\vspace{3pt}
\noindent\textbf{Impact of sidechain committee size.} In this experiment, we evaluate the overhead of consensus with committee sizes $S_c \in \{100, 250, 500, 750, 1000 \}$. We conduct our measurements over 10 sidechain rounds.

\begin{table}[h!]
    \caption{Impact of the committee size on consensus}
    \label{tab:com-size}
    \centering
    \begin{tabular}{|l|c|c|c|c|c|}
    \hline
        $S_c$ & 100 & 250 & 500 & 750& 1000   \\
    \hline
   Agreement time (s)   & 0.99 & 2.95  & 6.51 & 14.32 &  22.24 \\
   \hline
    \end{tabular}

\end{table}

A committee size impacts the duration needed to finalize agreement on a block due to the increased communication overhead. Our results in Table~\ref{tab:com-size} confirms this trend. In turn, this means that system designers should account for this increased agreement duration when setting the round duration for the sidechain, e.g., with $S_c = 1000$ a round should last at least for around 23 s and so on.

\end{document}